\documentclass[pra,twocolumn,superscriptaddress]{revtex4}

\usepackage{braket}
\usepackage{tikz}
\usepackage{mathtools}

\usepackage[caption=false]{subfig}

\usepackage{amsmath} 
                           
\usepackage{amsfonts}

\usepackage{amssymb}%
\usepackage[colorlinks=true,linkcolor=blue,citecolor=red,plainpages=false,pdfpagelabels]{hyperref}

\usepackage{graphicx}

\newtheorem{theorem}{Theorem}

\newtheorem{proposition}[theorem]{Proposition}

\newenvironment{proof}[1][Proof]{\noindent\textbf{#1.} }{\ \rule{0.5em}{0.5em}}

\newcommand{\Tr}{\operatorname{Tr}}
\newcommand{\Nn}{\mathcal{N}^{(n)}}
\newcommand{\Mn}{\mathcal{M}^{(n)}}
\newcommand{\Sn}{\mathcal{S}^{(n)}}

\allowdisplaybreaks

\begin{document}

\title{Evaluating the Advantage of Adaptive Strategies for Quantum Channel Distinguishability} 
\author{Vishal Katariya}
\affiliation{Hearne Institute for Theoretical Physics, Department of Physics and Astronomy, and Center for Computation and Technology, Louisiana State University, Baton Rouge, Louisiana 70803, USA}
\author{Mark M. Wilde}
\affiliation{Hearne Institute for Theoretical Physics, Department of Physics and Astronomy, and Center for Computation and Technology, Louisiana State University, Baton Rouge, Louisiana 70803, USA}
\affiliation{Stanford Institute for Theoretical Physics, Stanford University, Stanford, California 94305, USA}
\date{\today}

\begin{abstract}
	Recently, the resource theory of asymmetric distinguishability for quantum strategies was introduced  by [Wang \textit{et al}., Phys.~Rev.~Research~1, 033169 (2019)]. The fundamental objects in the resource theory are pairs of quantum strategies, which are generalizations of quantum channels that provide a framework to describe an arbitrary quantum interaction. In the present paper, we provide semi-definite program characterizations of the one-shot operational quantities in this resource theory. We then apply these semi-definite programs to study the advantage conferred by adaptive strategies in discrimination and distinguishability distillation of generalized amplitude damping channels. We find that there are significant gaps between what can be accomplished with an adaptive strategy versus a non-adaptive strategy.
\end{abstract}

\maketitle

\section{Introduction}

In quantum information theory, the tasks of quantum state and channel discrimination have been studied in a considerable amount of detail; see Refs.~\cite{H69,H73,Hel76,HP91,ON00} and \cite{Kitaev1997,AKN98,CPR00,Acin01,RW05,GLN04}, respectively. Given the central importance of distinguishing quantum states or channels, it is reasonable to study distinguishability itself in the context of a resource theory \cite{Matsumoto2010,Wang2019b, Wang2019a}, i.e., to use resource-theoretic tools to quantify distinguishability, and to use these tools to study the tasks of distilling distinguishability from a pair of objects, diluting canonical units of distinguishability to a desired pair, and transforming one pair of entities to another pair.

References~\cite{Matsumoto2010,Wang2019b, Wang2019a} developed  some basic tools and a framework for the resource theory of asymmetric distinguishability. In some sense, the resource theory of asymmetric distinguishability can be thought of as a ``meta''-resource theory. The basic objects in this resource theory come in pairs, and their worth is decided by the distinguishability of the entities in a pair. This resource theory is also unique in the sense that all physical operations acting on each element of the pair are free. A variety of resource theories can be thought of as being derived from the resource theory of asymmetric distinguishability, by setting specific restrictions on the states or channels allowed for free \cite{Wang2019b}.

The most general discrimination task in quantum information theory is not that of discriminating channels, but that of distinguishing what are known as quantum strategies \cite{Gutoski2007, Gutoski2010, Gutoski2012}, also known as quantum combs, memory channels, or higher-order quantum maps \cite{Chiribella2008a,Chiribella2008, Chiribella2009}. A quantum strategy completely represents the actions of an agent in a multi-round interaction with another party, and forms the next rung in the hierarchical ladder that begins with quantum states and channels. A key insight of \cite{Chiribella2009} is that the hierarchy consisting of states, channels, superchannels, etc., ends with quantum strategies. That is, all so-called ``higher-order'' dynamics can be described as quantum strategies. Given this importance of quantum strategies, and the flexibility and power offered by the resource theory of asymmetric distinguishability, it is worthwhile to continue the study of it for quantum strategies, as initiated in Ref.~\cite{Wang2019a}.

In this paper, we provide several contributions to the resource theory of asymmetric distinguishability  \cite{Wang2019b, Wang2019a}. 
One of our main contributions is a semi-definite programming (SDP) characterization of two crucial quantities in this resource theory: the one-shot distillable distinguishability and the one-shot distinguishability cost of quantum strategies, which characterize the resource theory's distillation and dilution tasks, respectively. To do so, we build upon the previous SDP characterizations of the quantum strategy distance \cite{Chiribella2008,Gutoski2012}, which provides a distance measure between strategies.

The other main contribution of this paper is to apply these SDPs to study particular examples of channel distinguishability tasks. As indicated in Ref.~\cite{Wang2019a}, distinguishability distillation is closely linked to asymmetric quantum channel discrimination. In quantum channel discrimination, one can employ either parallel or adaptive strategies. By definition, adaptive strategies are no less powerful than parallel ones. It is known that in the asymptotic limit, adaptive strategies confer no advantage over parallel ones in asymmetric channel discrimination \cite{Hayashi2009,Berta2018b,Fang2019}. This leaves open the question of whether adaptive strategies can help in channel discrimination when a finite number of channel uses are allowed. Our SDP formulations help us compute and study this gap. As an example, we consider distinguishability tasks involving generalized amplitude damping channels (GADCs) and show that adaptive strategies offer a significant advantage over parallel ones with respect to various distinguishability metrics of interest, thus extending prior work on this topic from Ref.~\cite{Harrow2010}.


\section{Quantum Strategies}

The idea of quantum strategies, combs, or higher-order maps, goes back over a decade  \cite{Gutoski2007,Chiribella2009}. A quantum strategy generalizes a quantum channel, in that it allows for sequential interactions over multiple rounds. Consider that there are two parties Alice and Bob. Alice's $n$-turn quantum strategy describes her actions in an $n$-round interaction with Bob. In such a scenario, Bob's $n$-round interaction is described by a suitable quantum co-strategy. In other words, the interaction of Alice's $n$-turn quantum strategy with another suitable $n$-turn strategy (belonging to Bob) captures all possible interactive behavior that takes place over $n$ rounds between them. 
Reference~\cite{Chiribella2009} introduced the term ``quantum comb,'' which refers to the same physical object as a quantum strategy. 

An $n$-turn quantum strategy $\Nn$, with $n \geq 1$, input systems $A_1$ through $A_n$, and output systems $B_1$ through $B_n$, consists of the following: (a) memory systems $M_1$ through $M_{n-1}$, and (b) quantum channels $\mathcal{N}_{A_{1}\rightarrow M_{1}B_{1}}^{1}$, $\mathcal{N}_{M_{1}A_{2}\rightarrow M_{2}B_{2}}^{2}$, \ldots, $\mathcal{N}_{M_{n-2} A_{n-1}\rightarrow M_{n-1}B_{n-1}}^{n-1}$, and $\mathcal{N}_{M_{n-1} A_{n}\rightarrow B_{n}}^{n}$. The definition above allows for any of the input, output, or memory systems to be trivial, which means that state preparation and measurements can be captured in the framework of quantum strategies. For the sake of brevity, we use the notation $A^n$ to denote systems $A_1$ through $A_n$. Figure~\ref{fig:strategy-costrategy-figure} depicts a three-turn strategy interacting with a three-turn co-strategy.
\begin{figure}
	\centering
	\includegraphics{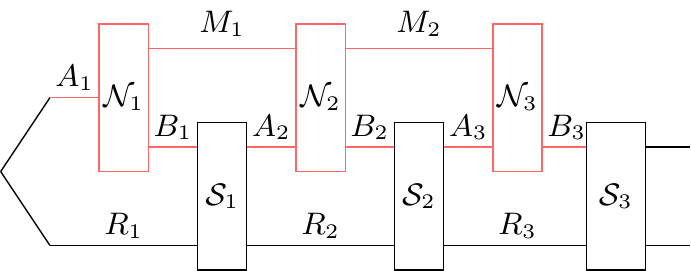}
	\caption{A three-turn strategy $\mathcal{N}^{(3)}$ interacts with a three-turn co-strategy $\mathcal{S}^{(3)}$. In red is the entirety of the three-turn strategy~$\mathcal{N}^{(3)}$, which consists of three quantum channels $\mathcal{N}_1$ through $\mathcal{N}_3$, connected to each other via memory systems. The co-strategy $\mathcal{S}^3$ consists of an initial state on systems $A_1 R_1$ as well, the quantum channels $\mathcal{S}_1$, $\mathcal{S}_2$, and $\mathcal{S}_3$, and memory systems $R_1$, $R_2$, and $R_3$.}
	\label{fig:strategy-costrategy-figure}
\end{figure}

A superchannel is a physical operation that converts one quantum channel to another \cite{Chiribella2008a, Gour2019}. It is a particular type of quantum strategy. Reference~\cite{Chiribella2008a} made the important observation that a superchannel can be equivalently represented as a bipartite channel, along with a causality constraint that defines the causal order of inputs and outputs. Reference~\cite{Chiribella2009}'s observation that quantum combs are all that are needed to describe higher-order quantum dynamics ties in neatly with, and generalizes, the superchannel-bipartite channel isomorphism. A superchannel can be cast as a bipartite channel, and likewise an object that transforms superchannels to superchannels (which is a quantum strategy) is itself a multipartite superchannel, which by the previously stated isomorphism is a multipartite channel~\cite{Chiribella2009}. Therefore, there is a ``collapse'' of the hierarchy that proves to be important, which implies that all higher-order quantum dynamics can be studied using the framework of quantum strategies \cite{Chiribella2009}.

Another isomorphism that is crucial in quantum information is the Choi isomorphism. It too establishes an equivalence between two different classes of objects--a single-party quantum channel can be equivalently represented by a bipartite quantum state. Putting the pieces together, we see that one can define a Choi state, or a Choi operator, not only for quantum channels, but also in general for quantum strategies. This isomorphism enables us to apply the tools developed in the resource theory of asymmetric distinguishability for states and channels to superchannels and, more generally, to quantum strategies. This was identified and studied in Ref. \cite{Wang2019a}, and here we elaborate in much more detail on these points.

In the remainder of this section, we establish some preliminaries regarding the quantum strategies formalism, and we also provide a semi-definite program for the quantum strategy distance between two $n$-round strategies that is slightly different from that presented in Ref.~\cite{Gutoski2012}.


\subsection{Choi operator and causality constraints}

To establish the Choi operator for a quantum strategy, we recall that a superchannel $\Theta_{(A_1 \rightarrow B_1) \rightarrow (A_2 \rightarrow B_2)}$ transforming $\mathcal{N}_{A_1 \rightarrow B_1}$ to $\mathcal{K}_{A_2 \rightarrow B_2}$ is in one-to-one correspondence with a bipartite channel $\mathcal{L}_{A_2 B_1 \rightarrow A_1 B_2}$ that has a certain no-signaling constraint \cite{Chiribella2008a, Gour2019}. The superchannel $\Theta_{(A_1 \rightarrow B_1) \rightarrow (A_2 \rightarrow B_2)}$ can be implemented via pre-processing and post-processing channels $\mathcal{E}_{A_2 \rightarrow A_1 M}$ and $\mathcal{D}_{B_1 M \rightarrow B_2}$ that share a memory system $M$. The Choi operator of the superchannel $\Theta_{(A_1 \rightarrow B_1) \rightarrow (A_2 \rightarrow B_2)}$, given by $\Gamma^{\Theta}_{A_1 A_2 B_1 B_2}$, is identified with the Choi operator of the corresponding bipartite channel
\begin{equation}
\mathcal{L}_{A_2 B_1 \rightarrow A_1 B_2} \coloneqq  \mathcal{D}_{B_1 M \rightarrow B_2} \circ \mathcal{E}_{A_2 \rightarrow A_1 M},
\end{equation}
along with a causality constraint that ensures no backward signaling in time; i.e., the $A$ systems can signal to the $B$ systems, but not vice versa. This is mathematically represented as \begin{equation}
\Gamma^{\Theta}_{A_1 A_2 B_1} = \Gamma^{\Theta}_{A_1 A_2} \otimes \pi_{B_1},
\end{equation}
 where $\pi_{B_1}$ is the maximally mixed state. This reasoning can be extended to quantum strategies.

A general $n$-turn quantum strategy $\Nn: A^n \rightarrow B^n$ is uniquely associated to its Choi operator $\Gamma^{\Nn}_{A^n B^n}$ via \cite{Gutoski2007}
\begin{equation}
\Gamma^{\Nn}_{A^n B^n} \coloneqq  \Nn_{A'^n \rightarrow B^n} (  \Gamma_{A'_1 A_1} \otimes \Gamma_{A'_2 A_2} \otimes \dots \otimes \Gamma_{A'_{n} A_n} )
\end{equation}
where $\Gamma_{A' A}\equiv |\Gamma\rangle\!\langle \Gamma|_{A'A}$ and $|\Gamma\rangle_{A'A} = \sum_{i} |i\rangle_{A'} |i\rangle_{A}$ is the unnormalized maximally entangled vector on systems $A' A$. The constraints on the Choi operator $\Gamma^{\Nn}_{A^n B^n}$ are that
\begin{equation}
\Gamma^{\mathcal{N}^{(n)}}_{ A^n  B^n} \geq 0,
\end{equation}
and that there exist $n$ positive semi-definite operators $N_{[1]}$, $N_{[2]},$ \ldots, $N_{[n]}$, with $N_{[i]}$ acting on systems $A^i B^i  $ for $1 \leq i \leq n$, such that
\begin{equation}
	\begin{aligned}
	N_{[n]} & = \Gamma^{\Nn}_{A^n B^n}, \\
		\Tr_{B_n} [ N_{[n]} ] &= N_{[n-1]} \otimes I_{A_n}, \\
		\Tr_{B_{n-1}} [ N_{[n-1]} ] &= N_{[n-2]} \otimes I_{A_{n-1}}, \\
		\vdots \\
		\Tr_{B_{2}} [ N_{[2]} ] &= N_{[1]} \otimes I_{A_2},  \\
		\Tr_{B_1} [ N_{[1]} ] &= I_{A_1}.
	\end{aligned} \label{eq:strategy-constraints}
\end{equation}

These latter constraints \eqref{eq:strategy-constraints} are causality constraints that arise due to the flow of information in the strategy. They generalize the single causality constraint imposed on the Choi operator of a superchannel. Conversely, if an operator $\Gamma^{\Nn}_{A^n B^n}$ satisfies the above constraints, then there is a quantum strategy associated to it \cite{Gutoski2007}. 

\subsection{Link Product}

How do we ``connect'' or compose two quantum strategies? To this end, the notion of link product was introduced to denote the composition, or interaction, of two quantum strategies \cite{Chiribella2009}. Two quantum strategies are composed by connecting the appropriate input and output systems, with an example being given in Fig.~\ref{fig:strategy-costrategy-figure}. Suppose that $n$-turn strategy $\Nn$ takes systems $A^n$ to $B^n$ and $m$-turn strategy $\mathcal{S}^{(m)}$ takes systems $C^m$ to  $D^m$. The Choi operator of the composition $\Nn \circ \mathcal{S}^{(m)}$ is given by $\Gamma^{\Nn} * \Gamma^{\mathcal{S}^{(m)}}$, defined in \eqref{eq:link-product}. Here, the nomenclature ``comb'' shines, as we connect the two strategies as if they were interlocking pieces, making sure to connect the appropriate input and output ports of the first and second strategy, respectively.

Qualitatively, the link product connects and ``collapses'' matching input and output systems of the two strategies. The composition $\Nn \circ \mathcal{S}^{(m)}$ is another strategy that takes systems 
\begin{equation}
\left(A^n \setminus D^m\right) \left( C^m \setminus B^n \right) \rightarrow \left( D^m \setminus A^n \right) \left( B^n \setminus C^m \right).
\end{equation}
The matching systems in this case are $A^n \cap D^m$ and $B^n \cap C^m$. To maintain brevity, we define
\begin{equation}
N \cap S \coloneqq  \left( A^n \cap D^m \right) \cup \left( B^n \cap C^m\right),
\end{equation}
as well as
\begin{equation}
N \setminus S \coloneqq  \left( A^n B^n\right) \setminus \left(C^m D^m \right)
\end{equation}
and
\begin{equation} 
S \setminus N =  \left( C^m D^m \right) \setminus \left( A^n B^n \right).
\end{equation}
The Choi operator of the composition $\Nn \circ \Sn$ is given by the link product of strategy Choi operators $\Gamma^{\Nn}_{A^n B^n}$ and $\Gamma^{\mathcal{S}^{(m)}}_{C^m D^m}$, and is defined as follows: 
\begingroup\allowdisplaybreaks[0]
\begin{multline} 
	\Gamma^{\Nn}_{A^n B^n} * \Gamma^{\mathcal{S}^{(m)}}_{C^m D^m} \coloneqq  \\
	\Tr_{N \cap S} \left[ \left(I_{N \setminus S} \otimes (\Gamma^{\mathcal{S}})^{T_{N \cap S}} \right) \left( \Gamma^{\mathcal{N}} \otimes I_{S \setminus N} \right) \right]
	\label{eq:link-product}
\end{multline}
\endgroup
where the notation $T_{N \cap S}$ refers to taking the partial transpose on systems $N \cap S$.

\subsection{Telling two strategies apart}

\label{subsec:strategy-distance}

It is natural to introduce a notion of distance, or distinguishability, between two strategies. In this vein, the quantities quantum strategy distance \cite{Chiribella2008, Chiribella2009, Gutoski2012}, the strategy fidelity \cite{Gutoski2018}, and the strategy max-relative entropy \cite{Chiribella2016} were previously defined. These are generalized by the generalized strategy divergence of Ref.~\cite{Wang2019a}.

Given two $n$-turn strategies with the same input and output systems, the most general discrimination strategy is defined analogously to that in channel discrimination; instead of passing a common state to two channels, one interacts a common $n$-turn co-strategy with the unknown strategy to obtain an output state on which a measurement is performed. That is, for strategies $\Nn$ and $\Mn: A^n \rightarrow B^n$, consider an arbitrary $n$-turn co-strategy $\Sn: B^{n-1} \rightarrow A^{n} R_n$. The compositions $\Nn \circ \Sn$ and $\Mn \circ \Sn$ yield states on $R_n B_n$. The strategy distance between $\Nn$ and $\Mn$ is the maximum trace distance between the states on $R_n B_n$ corresponding to strategies $\Nn$ and $\Mn$:
\begin{equation} \label{eq:strategy-distance}
\left\Vert \Nn - \Mn \right\Vert_{\diamond n} \coloneqq  \sup_{\Sn} \left\Vert \Nn \circ \Sn - \Mn \circ \Sn \right\Vert_{1}.
\end{equation}

The quantum strategy distance denotes the maximum classical trace distance between the output probability distributions produced by processing both strategies with a common co-strategy. For two arbitrary $n$-turn strategies, the strategy distance can be computed via a semi-definite program (SDP)~\cite{Gutoski2012}, which provides a powerful tool that can be used to study, among other things, the advantage provided by adaptive strategies over parallel ones in quantum channel discrimination, explored in Section~\ref{sec:adaptive-channel-discrimination}. 

In what follows, we present an SDP for the \textit{normalized} quantum strategy distance $\frac{1}{2}\left\Vert \Nn - \Mn \right\Vert_{\diamond n}$ of two strategies  that is slightly different from that presented previously, in Ref.~\cite{Gutoski2012}. This alternate form of the strategy distance is used later to derive SDP characterizations of the distillable distinguishability and the distinguishability cost in Section~\ref{sec:sdp-quantities}. 

\begin{proposition}
The normalized strategy distance $\frac{1}{2}\left\Vert \Nn - \Mn \right\Vert_{\diamond n}$ can be expressed as the following SDP, where $\Gamma^{\mathcal{N}^{(n)}}_{A^n B^n}$ and $\Gamma^{\mathcal{M}^{(n)}}_{A^n B^n}$ are the Choi operators of the strategies $\Nn_{A^n \rightarrow B^n}$ and $\Mn_{A^n \rightarrow B^n}$:
\begin{equation}
\sup_{S, S_{[n]}, \cdots, S_{[1]} \geq0}\left\{
\begin{array}
[c]{c}%
\Tr[ S (\Gamma^{\Nn}  - \Gamma^{\Mn})   ]:\\
\begin{aligned}
S & \leq S_{[n]} \otimes I_{B_n}, \\
\Tr_{A_n} [ S_{[n]} ] & = S_{[n-1]} \otimes I_{B_{n-1}}, \\
\vdots\\
\Tr_{A_2} [ S_{[2]} ] & = S_{[1]} \otimes I_{B_{1}}, \\
\Tr [ S_{[1]} ] & = 1
\end{aligned}
\end{array}
\right\}  .
\label{eq:primal-strategy-dist}
\end{equation}
The dual of the normalized strategy distance is
\begin{equation}
\inf_{\substack{\mu \in \mathbb{R},Y_n \geq 0, \\ Y_1, \dots, Y_{n-1} \in \text{Herm}}}\left\{
\begin{array}
[c]{c}%
\mu : \\
\begin{aligned}
Y_n & \geq \Gamma^{\Nn} - \Gamma^{\Mn},\\
Y_{n-1} \otimes I_{A_n} & \geq \Tr_{B_n}[ Y_n ] ,\\
Y_{n-2} \otimes I_{A_{n-1}} & \geq \Tr_{B_{n-1}}[ Y_{n-1}] ,\\
\vdots\\
Y_{1} \otimes I_{A_{2}} & \geq \Tr_{B_{2}}[ Y_{2}] ,\\
\mu I_{A_1} & \geq \Tr_{B_1}[ Y_{1}]
\end{aligned}
\end{array}
\right\}  . \label{eq:dual-strategy-distance}
\end{equation} 
\end{proposition}

\begin{proof}
	We start by recalling the SDP formulation of the  strategy distance from Ref.~\cite{Gutoski2012}:
	\begin{multline}
		\left\Vert \Nn - \Mn \right\Vert_{\diamond n} =
		\\ \sup_{T_0, T_1 \geq0} \left\{ \begin{array}
			[c]{c}%
			\Tr[  (\Gamma^{\Nn}  - \Gamma^{\Mn})  (T_0 - T_1) ]:\\
			\begin{aligned}
				T_0 + T_1 & = T_{[n]} \otimes I_{B_n}, \\
				\Tr_{A_n} [ T_{[n]} ] & = T_{[n-1]} \otimes I_{B_{n-1}}, \\
				\vdots\\
				\Tr_{A_2} [ T_{[2]} ] & = T_{[1]} \otimes I_{B_{1}}, \\
				\Tr [ T_{[1]} ] & = 1
			\end{aligned}
		\end{array}
		\right\}.
		\label{eq:GW-orig-SDP}
	\end{multline}
	In the above, $\{ T_0, T_1\}$ comprise the Choi operators of an $n$-round measuring co-strategy, as defined in Ref.~\cite{Gutoski2012}. This means that $T = T_0 + T_1$ is the Choi operator of an $n$-round non-measuring co-strategy and obeys the following constraints:
	\begin{align}
		T &\geq 0 \\
		T &= T_{[n]} \otimes I_{B_n}, \\
		\Tr_{A_n} [ T_{[n]} ] & = T_{[n-1]} \otimes I_{B_{n-1}}, \\
		\vdots\\
		\Tr_{A_2} [ T_{[2]} ] & = T_{[1]} \otimes I_{B_{1}}, \\
		\Tr [ T_{[1]} ] & = 1.
	\end{align} 
	The objective function
	\begin{equation}
		\Tr[  (\Gamma^{\Nn}  - \Gamma^{\Mn})  (T_0 - T_1) ]
	\end{equation}
	can be rewritten as follows:
	\begin{align}
		& ~\Tr \!\left[ \left( \Gamma^{\mathcal{N}^{(n)}} - \Gamma^{\mathcal{M}^{(n)}} \right) (T_0 - T_1) \right]\notag \\
		= &~ \Tr \!\left[ \left( \Gamma^{\mathcal{N}^{(n)}} - \Gamma^{\mathcal{M}^{(n)}} \right) \left(T_0 - (T- T_0)\right) \right] \notag\\
		= &~ 2 \Tr \!\left[ \left( \Gamma^{\mathcal{N}^{(n)}} - \Gamma^{\mathcal{M}^{(n)}} \right) T_0 \right]- \Tr \!\left[ \left( \Gamma^{\mathcal{N}^{(n)}} - \Gamma^{\mathcal{M}^{(n)}} \right) T \right] \notag\\
		= &~ 2 \Tr \!\left[ \left( \Gamma^{\mathcal{N}^{(n)}} - \Gamma^{\mathcal{M}^{(n)}} \right) T_0 \right]. \label{eq:app-rewriting-gutoski}
	\end{align}
	In the above, the first equality arises because $T = T_0 + T_1$. The third equality is due to the fact that for the strategy $\Nn$, the probabilities of obtaining outcomes corresponding to $T_0$ and $T_1$ add to 1, i.e.,
	\begin{equation}
	\Tr[ \Gamma^{\mathcal{N}^{(n)}} (T_0 + T_1) ] = \Tr[ \Gamma^{\mathcal{N}^{(n)}} T ] = 1.
	\end{equation}
	The same holds for $\Mn$ as well, and thus $\Tr \!\left[ \left( \Gamma^{\mathcal{N}^{(n)}} - \Gamma^{\mathcal{M}^{(n)}} \right) T \right] = 0$.
	
	The operator $T_0$ is the Choi operator of an $n$-round measuring co-strategy, and so it obeys the following constraints:
	\begin{align}
		T_0 &\geq 0, \\
		T_0 & \leq T_{[n]} \otimes I_{B_n}, \\
		\Tr_{A_n} [ T_{[n]} ] & = T_{[n-1]} \otimes I_{B_{n-1}}, \\
		\vdots\\
		\Tr_{A_2} [ T_{[2]} ] & = T_{[1]} \otimes I_{B_{1}}, \\
		\Tr [ T_{[1]} ] & = 1.
	\end{align} 	
	By combining \eqref{eq:app-rewriting-gutoski} with the constraints on $T_0$, we arrive at the desired semi-definite program in \eqref{eq:primal-strategy-dist}. Note also that we can start with \eqref{eq:primal-strategy-dist} and run the whole proof backwards to arrive at \eqref{eq:GW-orig-SDP}, setting $T_0=S$, $T_1 = S_{[n]} \otimes I_{B_n} - S$, and $T = S_{[n]} \otimes I_{B_n}$.

	The dual is given by \eqref{eq:dual-strategy-distance}, which can be verified by the Lagrange multiplier method.
\end{proof}

\section{Distinguishability Resource Theory}

We first recall some aspects of the resource theory of asymmetric distinguishability, work on which was begun in~Ref.~\cite{Matsumoto2010} and continued in Refs.~\cite{Wang2019b, Wang2019a, Rethinasamy2019}.
The objects in consideration in this resource theory are pairs of like objects. These objects can be probability distributions, quantum states, quantum channels, or most generally, quantum strategies of an equal number of rounds. Any operation on the pair elements is considered free, justified by the fact that data processing cannot increase the distinguishability of two objects.

The object $(\rho, \sigma)$, a state box, is an ordered pair of states that is to be understood as an atomic entity: upon being handed a state box, one does not know which state it contains. In this paper, we consider ordered pairs of $n$-turn quantum strategies, which generally are represented by $( \Nn, \Mn )$.

\subsection{Bits of asymmetric distinguishability}

Here, we recall the canonical unit of asymmetric distinguishability (AD) \cite{Wang2019b}. The state box $\left( |0\rangle\!\langle 0|, \pi \right)$ encapsulates one bit of AD, where
\begin{equation}
\pi \coloneqq  \frac{1}{2} \left( |0\rangle\!\langle 0| + |1\rangle\!\langle 1| \right)
\end{equation}
 is the maximally mixed qubit state. Defining this unit enables us to quantify the amount of resource present in an arbitrary strategy box. As discussed in Ref.~\cite{Wang2019b}, the bit of AD represents a pair of experiments in which the null hypothesis corresponds to preparing  $|0\rangle\!\langle 0|$, and the alternative hypothesis corresponds to preparing $\pi$.  A number $m$ bits of asymmetric distinguishability corresponds to the box $\left( |0\rangle\!\langle 0|^{\otimes m}, \pi^{\otimes m} \right)$. Alternatively, the state box $\left( |0\rangle\!\langle 0|, \pi_M \right)$, with
\begin{equation} 
 \pi_M \coloneqq  \frac{1}{M} |0\rangle\!\langle 0| + \left( 1 - \frac{1}{M} \right) |1\rangle\!\langle 1| ,
 \end{equation}
 contains $\log_2 M$ bits of~AD. 

\subsection{Distillation and dilution of strategy boxes}

Given a strategy box $( \Nn, \Mn )$, we are interested in two questions: (a) how many bits of AD can be distilled from it, and (b) how many bits of AD are required so that one can dilute them to $( \Nn, \Mn )$? These quantities are crucial to the resource theory of asymmetric distinguishability. The one-shot versions of these tasks are explained below, and we also provide  explicit semi-definite programs for them.

We start by discussing distinguishability distillation. The goal of approximate distillation is to transform a strategy box into as many approximate bits of AD as possible. Quantitatively, the one-shot $\varepsilon$-approximate distillable distinguishability of strategy box $( \Nn, \Mn )$ is given by
\begingroup\allowdisplaybreaks[0]
\begin{multline} 
D_{d}^{\varepsilon}(\mathcal{N}^{(n)},\mathcal{M}^{(n)})\coloneqq  \\
\log_{2}\sup_{\mathcal{S}^{(n)}}\{M:\Nn \circ \Sn \approx_{\varepsilon} |0\rangle\!\langle 0|, \Mn \circ \Sn =\pi_{M}\} \notag ,
\end{multline}
\endgroup
where $\Sn$ is an $n$-turn co-strategy that interacts with $\Nn$ and $\Mn$ to yield a qubit state, and 
\begin{equation}
	\Nn \circ \Sn \approx_{\varepsilon} |0\rangle\!\langle 0| \Longleftrightarrow \frac{1}{2} \left\Vert \Nn \circ \Sn - |0\rangle\!\langle 0| \right\Vert_{1} \leq \varepsilon.
\end{equation}


The operational quantity for approximate distillation is the smooth strategy min-relative entropy. The smooth strategy min-relative entropy between $n$-turn strategies $\Nn$ and $\Mn$ is defined as follows: 
\begin{equation}
	D_{\min}^{\varepsilon} ( \Nn \Vert \Mn ) \coloneqq  \sup_{\Sn} D_{\min}^{\varepsilon} ( \Nn \circ \Sn \Vert \Mn \circ \Sn ) \notag
\end{equation}
where $\Nn$ and $\Mn$ take systems $A^n$ to $B^n$, and $\Sn$ is an $n$-turn co-strategy that takes systems $B^{n-1}$ to $A^n R_n$. The smooth min-relative entropy of states is defined as \cite{Buscemi2010, Brandao2011, Wang2012}
\begin{equation}
	D_{\min}^{\varepsilon}(\rho\Vert\sigma)\coloneqq -\log_{2}\inf_{0\leq\Lambda\leq I}\left\{  \operatorname{Tr}[\Lambda\sigma]:\operatorname{Tr}[\Lambda\rho]\geq1-\varepsilon\right\} \notag .
\end{equation}

Distinguishability dilution, on the other hand, is the complementary task to distillation. Approximate dilution refers to the task of transforming $\left( |0\rangle\!\langle 0|, \pi_M \right)$ to approximately one copy of $( \Nn, \Mn )$ with as small $M$ as possible. The one-shot $\varepsilon$-approximate distinguishability cost of the box $( \Nn, \Mn )$ is given by the following:
\begingroup\allowdisplaybreaks[0]
\begin{multline}
D_{c}^{\varepsilon}(\mathcal{N}^{(n)},\mathcal{M}^{(n)})\coloneqq  \\
\log_{2}\inf_{\Sn}\{M:\Sn(|0\rangle\!\langle 0|)\approx_{\varepsilon}\mathcal{N}^{(n)}, \Sn (\pi_{M})=\mathcal{M}^{(n)}\}
\end{multline}
\endgroup
where
\begin{multline}
 \Sn(|0\rangle\!\langle 0|) \approx_{\varepsilon} \Nn \quad \Longleftrightarrow \\
 \quad \tfrac{1}{2} \left\Vert \Sn(|0\rangle\!\langle 0|) - \Sn \right\Vert_{\diamond n} \leq \varepsilon.
\end{multline}

For the dilution task, the operational quantity is the smooth strategy max-relative entropy of quantum channels and is defined as 
\begin{equation}
D_{\max}^{\varepsilon} ( \Nn \Vert \Mn ) \coloneqq   \inf_{\widetilde{\mathcal{N}}^{(n)} \approx_{\varepsilon} \Nn} D_{\max}( \widetilde{\mathcal{N}}^{(n)}  \Vert \Mn  ) ,
\end{equation}
where $D_{\max}( \widetilde{\mathcal{N}}^{(n)}  \Vert \Mn  )$ is equal to the max-relative entropy for strategies, defined as \cite{Chiribella2016}
\begin{equation}
D_{\max}( \widetilde{\mathcal{N}}^{(n)}  \Vert \Mn  ) \coloneqq  D_{\max}( \Gamma^{\widetilde{\mathcal{N}}^{(n)}}  \Vert \Gamma^{\Mn}  ) 
\label{eq:max-rel-ent-exact-strat},
\end{equation}
and the max-relative entropy for states is defined as $D_{\max}(\rho \Vert \sigma) \coloneqq  \inf\{\lambda : \rho \leq 2^\lambda \sigma\}$ \cite{Datta2009}.

We now state a result claimed in Ref.~\cite{Wang2019a}. Its detailed proof is given in Appendix~\ref{sec:one-shot-theorem-proof}.
\begin{theorem} \label{thm:approx-one-shot}
	As seen in Ref.~\cite{Wang2019a}, the approximate one-shot distillable distinguishability of the strategy box $\left(\Nn,\Mn\right)$ is equal to the smooth strategy min-relative entropy:
\begin{equation}
		D_{d}^{\varepsilon}(\Nn, \Mn)=D_{\min}^{\varepsilon}(\Nn \Vert \Mn ), 
\end{equation}
	and the approximate one-shot distinguishability cost is equal to the smooth strategy max-relative entropy:
\begin{equation}
		D_{c}^{\varepsilon}(\Nn, \Mn )=D_{\max}^{\varepsilon}(\Nn \Vert \Mn). 
\end{equation}
\end{theorem}


\subsection{SDPs for one-shot quantities}

\label{sec:sdp-quantities}

This section contains one of our main contributions: explicit semi-definite programs to calculate, for a given strategy box, the approximate distillable distinguishability and approximate distinguishability cost. 

\begin{proposition}
Considering strategies $\Nn$ and $\Mn$ to take systems $A^n$ to $B^n$, the distillable distinguishability is computable via the following semi-definite program:
\begin{multline}
\label{eq:strategy-min-entropy-primal}
2^{- D_{\min}^{\varepsilon} ( \Nn \Vert \Mn )} =
\\ \inf_{S, S_{[n]}, \dots, S_{[1]} \geq 0}\left\{
\begin{array}
[c]{c}%
\operatorname{Tr}[S\, \Gamma^{\mathcal{M}^{(n)}}  ]:\\
\begin{aligned}
	\operatorname{Tr}[ S\, \Gamma^{\mathcal{N}^{(n)}} ] &\geq 1-\varepsilon,\\
	S &\leq S_{[n]} \otimes I_{B_n}, \\
	\Tr_{A_n} [ S_{[n]} ] &= S_{[n-1]} \otimes I_{B_{n-1}}, \\
						  & \vdots \\
	\Tr_{A_2} [ S_{[2]} ] &= S_{[1]} \otimes I_{B_1}, \\
	\Tr [ S_{[1]} ] &= 1
\end{aligned}
\end{array}
\right\}  ,
\end{multline}
with dual
\begin{multline}
\label{eq:strategy-min-entropy-dual}
\sup_{\substack{\mu_1, Y_n \geq 0, \\
\mu_2 \in \mathbb{R}, \\ Y_1, \dots, \\
Y_{n-1} \in \text{Herm} }} \left\{
	\begin{array}
		[c]{c}
		(1- \varepsilon) \mu_1 - \mu_2 :\\
		\begin{aligned}
			Y_n &\geq \mu_1 \Gamma^{\mathcal{N}^{(n)}}  - \Gamma^{\mathcal{M}^{(n)}}, \\
			Y_{n-1} \otimes I_{A_n} &\geq \Tr_{B_n} [ Y_n ], \\
			Y_{n-2} \otimes I_{A_{n-1}} &\geq \Tr_{B_{n-1}} [ Y_{n-1} ], \\
									&\vdots \\
Y_{1} \otimes I_{A_{2}} &\geq \Tr_{B_{2}} [ Y_{2} ], \\
			\mu_2 I_{A_1} &\geq \Tr_{B_1} [ Y_{1} ] 
	\end{aligned}
	\end{array}
\right\}.
\end{multline}
\end{proposition}

\begin{proof}
	We have
	\begin{align}
		D_d^{\varepsilon} ( \Nn, \Mn ) &= D_{\min}^{\varepsilon} ( \Nn \Vert \Mn ) \\
		&= \sup_{\Sn} D_{\min}^{\varepsilon} ( \Nn \circ \Sn \Vert \Mn \circ \Sn )
	\end{align}
	and
	\begin{equation}
		D_{\min}^{\varepsilon}(\rho \Vert \sigma) = \! \! - \log_2 \inf_{\Lambda \geq 0} \left\{ \Tr[\Lambda \sigma ]: \! \Tr[ \Lambda \rho ] \geq 1 - \varepsilon, \Lambda \leq I \right\}.
	\end{equation}
	We consider $\Sn : B^{n-1} \rightarrow A^n R_n$ to be a co-strategy, so that $\Nn \circ \Sn$ is a quantum state on $R_n B_n$. Let $\Lambda_{R_n B_n}$ be a measurement operator such that $\Tr[\Lambda_{R_n B_n} (\Nn \circ \Sn)]$ is a probability.
	We now have
	\begin{multline}
		D_{\min}^{\varepsilon} ( \Nn \Vert \Mn ) = \\
		-\log_{2}\inf_{\Sn}\left\{
		\begin{array}
			[c]{c}%
			\operatorname{Tr}[ S^T \Gamma^{\Mn}  ]:\\
			\operatorname{Tr}[ S^T \Gamma^{\Nn}]\geq1-\varepsilon
		\end{array}
		\right\}
	\end{multline}
	such that $S$ is the Choi operator of a valid ``sub co-strategy'' corresponding to $\Sn$ and $\Lambda$ and we have exploited the link product from \eqref{eq:link-product}. To write it out explicitly, we use the following constraints on the Choi operator of a sub co-strategy \cite[Section~2.3]{Gutoski2012}:
	\begin{align}
		0 \leq S &\leq S_{[n]} \otimes I_{B_n}, \\
		\Tr_{A_n} [ S_{[n]} ] &= S_{[n-1]} \otimes I_{B_n}, \\
		&\vdots \\
		\Tr [ S_{[1]} ] &= 1.
	\end{align}
	so that
	\begin{multline}
		2^{-D_{\min}^{\varepsilon} ( \Nn \Vert \Mn )} = \\
		\inf_{S, S_{[n]}, \dots, S_{[1]} \geq 0}\left\{
		\begin{array}
			[c]{c}%
			\operatorname{Tr}[S^T \Gamma^{\Mn} ]:\\
			\begin{aligned}
				\operatorname{Tr}[S^T \Gamma^{\Nn}]  &\geq1-\varepsilon,\\
				S &\leq S_{[n]} \otimes I_{B_n}, \\
				\Tr_{A_n} [ S_{[n]} ] &= S_{[n-1]} \otimes I_{B_n}, \\
				&\vdots \\
				\Tr [ S_{[1]} ] &= 1
			\end{aligned}
		\end{array}
		\right\}.
	\end{multline}
	Finally, the full transpose of $S$ corresponds to a legitimate sub co-strategy and since we are optimizing over all of them, we can remove the transpose in the optimization to arrive at \eqref{eq:strategy-min-entropy-primal}.
	
	The dual program is then given by \eqref{eq:strategy-min-entropy-dual},
	which can be verified by means of the Lagrange multiplier method.  The details of this calculation are provided in the Appendixes.
\end{proof}

\begin{proposition}
For strategies $\Nn$ and $\Mn$ taking systems $A^n$ to $B^n$, the distinguishability cost is computable via the following semi-definite program:
\begingroup\allowdisplaybreaks[0]
\begin{multline}
\label{eq:strategy-max-entropy-primal}
2^{D_{\max}^{\varepsilon} ( \Nn \Vert \Mn )} = \\
 \inf_{\substack{\lambda, Y_n, N \geq 0, \\
 N_{[n]}, \ldots, N_{[1]}\geq 0\\Y_1, \dots, Y_{n-1} \in \text{Herm}}} \left\{
\begin{array}
			[c]{c}%
			\lambda: \\
			\begin{aligned}
				N &\leq \lambda \Gamma^{\Mn}\\
			Y_n &\geq \Gamma^{\Nn} - N, \\
			Y_{n-1} \otimes I_{A_n} &\geq \Tr_{B_n} [ Y_n ], \\
			Y_{n-2} \otimes I_{A_{n-1}} &\geq \Tr_{B_{n-1}} [ Y_{n-1} ], \\
									&\vdots \\
			Y_{1} \otimes I_{A_2} &\geq \Tr_{B_2} [ Y_{2} ], \\
			\varepsilon I_{A_1} &\geq \Tr_{B_1} [ Y_{1} ], \\
			\Tr_{B_n}[ N ] &= N_{[n-1]} \otimes I_{A_n}, \\
			\!\! \Tr_{B_{n-1}} [ N_{[n-1]} ] &= N_{[n-2]} \otimes I_{A_{n-1}} \!\!, \\
												   &\vdots \\
			\Tr_{B_{2}} [ N_{[2]} ] &= N_{[1]} \otimes I_{A_2}, \\
			\Tr _{B_1 }[ N_{[1]} ] &= I_{A_1}
		\end{aligned}
		\end{array}
\right\},
\end{multline}
\endgroup
with dual
\begingroup\allowdisplaybreaks
\begin{equation} \label{eq:strategy-max-entropy-dual}
	\sup_{\substack{W_1, W_2, \dots, \\ W_{n+2} \geq 0, \\ X_1, \dots, X_n \in \text{Herm}}} \left\{
		\begin{array}
			[c]{c}%
			\Tr[ \Gamma^{\Nn} W_{n+1} - \varepsilon W_{1} + X_1 ] : \\
			\begin{aligned}
			\!\!\!\Tr [ W_{n+2} \Gamma^{\Mn} ] &\leq 1, \\
			\!\!\!\! W_{n+2} &\geq W_{n+1} + X_n \otimes I_{B_n}, \!\!\! \\
			W_n \otimes I_{B_n} &\geq W_{n+1} ,\\
			W_{n-1} \otimes I_{B_{n-1}} &\geq \Tr_{A_n} [ W_n ] ,\\
			W_{n-2} \otimes I_{B_{n-2}} &\geq \Tr_{A_{n-1}} [ W_{n-1} ] ,\\
									&\vdots \\
			W_{1} \otimes I_{B_1} &\geq \Tr_{A_2} [ W_{2} ] ,\\
			\Tr_{A_n} [X_n ] &\geq X_{n-1} \otimes I_{B_{n-1}}  ,\\
			\Tr_{A_{n-1}} [ X_{n-1} ] &\geq X_{n-2} \otimes I_{B_{n-2}} ,\\
								  &\vdots \\
			\Tr_{A_3} [ X_{3} ] &\geq X_{2} \otimes I_{B_{2}} ,\\
			\Tr_{A_2} [ X_{2} ] &\geq X_1 \otimes I_{B_1}
		\end{aligned}
		\end{array}
	\right\}.
\end{equation}
\endgroup
\end{proposition}

\begin{proof}
	Firstly, we have
	\begin{equation}
		D_{\max}^{\varepsilon} ( \Nn \Vert \Mn ) = \inf_{\widetilde{\mathcal{N}}^{(n)} \approx_{\varepsilon} \Nn} D_{\max} ( \widetilde{\mathcal{N}}^{(n)} \Vert \Mn )
		\label{eq:smooth-dmax-app-proof}
	\end{equation}
	and the dual of the normalized strategy distance from \eqref{eq:dual-strategy-distance}
	\begin{multline}
		\frac{1}{2} \Vert \Nn - \widetilde{\mathcal{N}}^{(n)} \Vert_{\diamond n} = \\
		\inf_{\substack{\mu \in \mathbb{R},Y_n \geq 0, \\ Y_1 \dots Y_{n-1} \in \text{Herm}}}\left\{
		\begin{array}
			[c]{c}%
			\mu : \\
			\begin{aligned}
				Y_n &\geq\Gamma^{\mathcal{N}^{(n)}}-\Gamma^{\widetilde{\mathcal{N}}^{(n)}},\\
				Y_{n-1} \otimes I_{A_n} &\geq \Tr_{B_n}[ Y_n ] ,\\
				Y_{n-2} \otimes I_{A_{n-1}} &\geq \Tr_{B_{n-1}}[ Y_{n-1}] ,\\
				&\vdots\\
				Y_{1} \otimes I_{A_{2}} &\geq \Tr_{B_{2}}[ Y_{2}] ,\\
				\mu I_{A_1} &\geq \Tr_{B_1}[ Y_{1}]
			\end{aligned}
		\end{array}
		\right\}.
		\label{eq:dual-str-dist-proof}
	\end{multline}
	For $\widetilde{\mathcal{N}}^{(n)}$ the optimizer in \eqref{eq:smooth-dmax-app-proof} and exploiting \eqref{eq:max-rel-ent-exact-strat}, we have
	\begin{equation}
		2^{D_{\max}^{\varepsilon} ( \Nn \Vert \Mn )} = 
		\inf_{\lambda \geq 0} \left\{
		\lambda: 
		\Gamma^{\widetilde{\mathcal{N}}^{(n)}} \leq \lambda \Gamma^{\Mn}
		\right\}.
	\end{equation}
	Now we combine these while also adding constraints that ensure that $\widetilde{\mathcal{N}}^{(n)}$ is a valid quantum strategy. Therefore, we use the constraints in \eqref{eq:strategy-constraints} and incorporate them into the optimization. Thus we get
	\begin{multline}
		2^{D_{\max}^{\varepsilon} ( \Nn \Vert \Mn ) }= \\
		\inf_{\substack{\lambda, N, Y_n\geq 0, \\ Y_1, \dots ,Y_{n-1}\in\operatorname{Herm} \\ N_{[1]}, \dots,  N_{[n-1]} \geq 0}} \left\{
		\begin{array}
			[c]{c}%
			\lambda: \\
			\begin{aligned}
				N &\leq \lambda \Gamma^{\Mn},\\
				Y_n &\geq \Gamma^{\Nn} - N ,\\
				Y_{n-1} \otimes I_{A_n} &\geq \Tr_{B_n} [Y_n ] ,\\
				Y_{n-2} \otimes I_{A_{n-1}} &\geq \Tr_{B_{n-1}} [ Y_{n-1} ] ,\\
				&\vdots \\
				Y_{1} \otimes I_{A_2} &\geq \Tr_{B_2} [ Y_{2} ] ,\\
				\varepsilon I_{A_1} &\geq \Tr_{B_1} [ Y_{1} ] ,\\
				\Tr_{B_n}[ N ] &= N_{[n-1]} \otimes I_{A_n} ,\\
				\!\! \Tr_{B_{n-1}} [ N_{[n-1]} ] &= N_{[n-2]} \otimes I_{A_{n-1}} , \!\!\! \\
				&\vdots \\
				\Tr_{B_{2}} [ N_{[2]} ] &= N_{[1]} \otimes I_{A_2} ,\\
				\Tr_{B_1} [ N_{[1]} ] &= I_{A_1}
			\end{aligned}
		\end{array}
		\right\}.\notag
	\end{multline}
	The dual is given by \eqref{eq:strategy-max-entropy-dual},
	which can be verified by means of the Lagrange multiplier method.
\end{proof}

These SDPs are extensions of those presented in \cite[Appendix~C-3]{Wang2019a}, with the difference being that the above ones incorporate causality constraints for quantum strategies. They enable us to efficiently compute these quantities for various scenarios of interest. We do so in the following, where we investigate whether adaptive strategies provide an advantage over parallel ones with respect to the quantum strategy distance, the approximate distillable distinguishability, and the approximate distinguishability cost.

\section{Adaptive vs.~non-adaptive in discrimination and distillation} \label{sec:adaptive-channel-discrimination}

\begin{figure}
\centering

\subfloat[Parallel strategy \label{fig:parallel-strategy}]{\includegraphics[width=0.5\columnwidth]{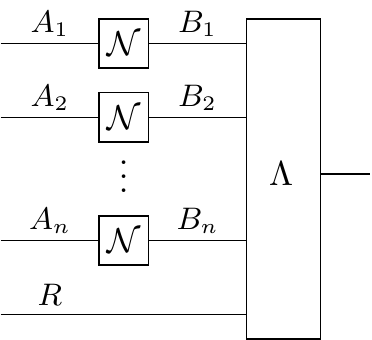}}\\ 
\subfloat[Adaptive strategy \label{fig:adaptive-strategy}]{\includegraphics[width=\columnwidth]{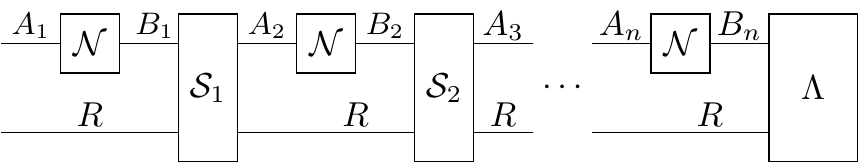}} 


\caption{Processing $n$ uses of a channel $\mathcal{N}_{A \rightarrow B}$ in (a) a parallel manner, and (b) an adaptive manner. In a parallel strategy, the input state to the $n$ copies of the channel can be entangled, but the channels are all applied simultaneously. In an adaptive strategy, the channels are applied sequentially, with interleaving channels $\mathcal{S}_1$ through $\mathcal{S}_{n-1}$ providing adaptive feedback. In a channel discrimination protocol, channels $\mathcal{N}_{A \rightarrow B}$ and $\mathcal{M}_{A \rightarrow B}$ are processed using a common apparatus and the final measurement results are used to decide the correct channel.}
\label{fig:parallel-adaptive-strategies}
\end{figure}

\begin{figure} \centering
\includegraphics[width=3.2in]{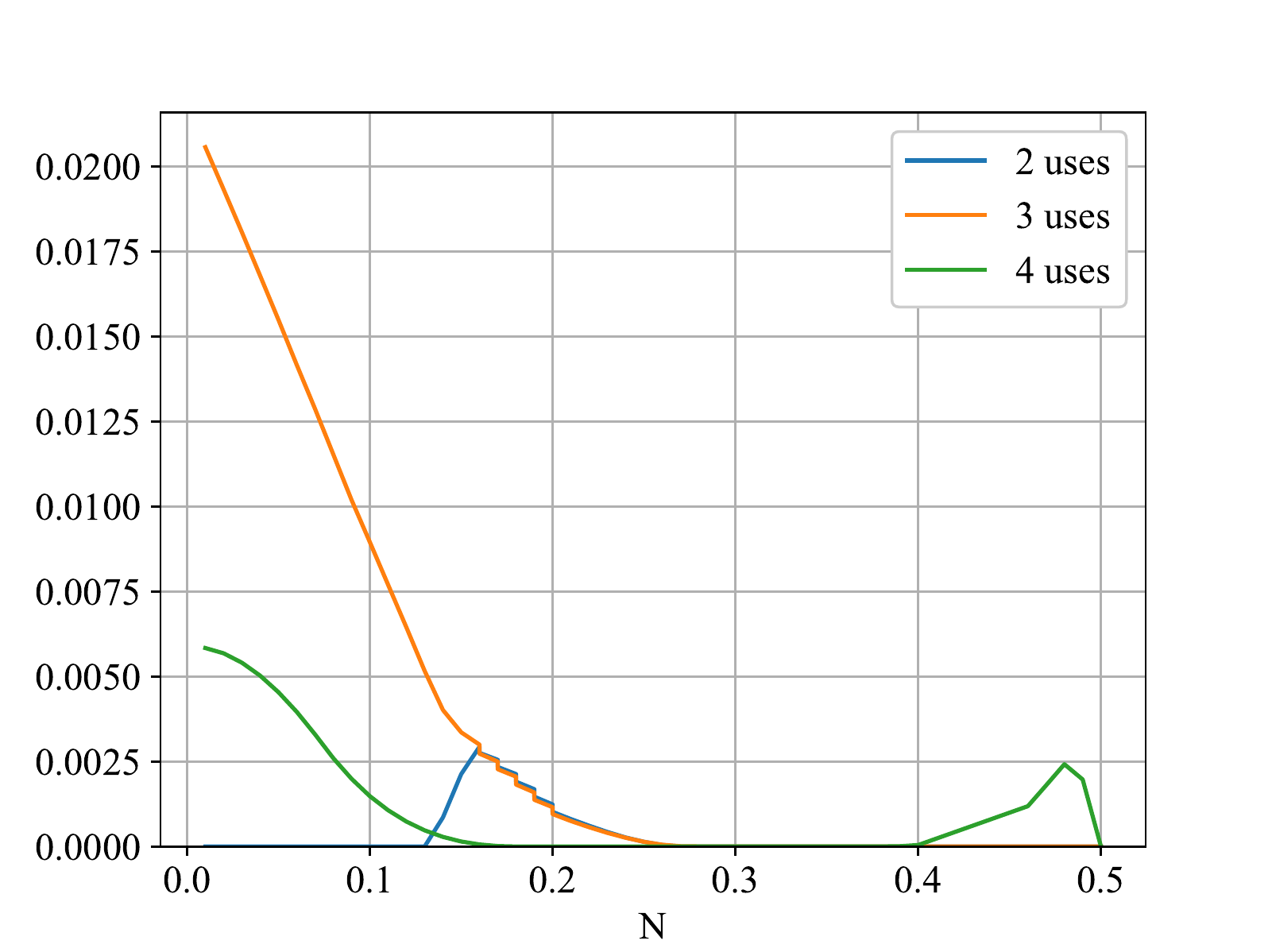}
\caption{Consider two GADCs with $\gamma = 0.2$ and $0.3$, respectively. We plot the \textit{difference} between $\frac{1}{2} \Vert \Nn - \Mn \Vert_{\diamond n}$ and $\frac{1}{2} \Vert \mathcal{N}^{\otimes n} - \mathcal{M}^{\otimes n} \Vert_{\diamond}$, where the strategies $\Nn$ and $\Mn$ each consist of $n$ instances of the same channel. While varying the common parameter $N$, and allowing for different number of channel uses, we see that adaptive strategies offer an advantage in discrimination over parallel ones.}
\label{fig:discrimination-difference-plot}
\end{figure}

Strategies that distinguish between two channels $\mathcal{N}_{A \rightarrow B}$ and $\mathcal{M}_{A \rightarrow B}$ using each channel $n$ times are adaptive in general. Parallel strategies are a special case of adaptive strategies that are of practical interest. Parallel strategies involve a distinguisher inputting a possibly entangled state simultaneously to $n$ instances of the unknown channel. Adaptive strategies, on the other hand, involve $n$ uses of the unknown channel that happen sequentially. Between uses of the unknown channel, the distinguisher can perform a quantum channel so as to boost the chances of success. These two scenarios are described in Figure~\ref{fig:parallel-adaptive-strategies}.

A parallel strategy is a special case of an adaptive strategy~\cite{Chiribella2008}. Adaptive strategies are therefore no less powerful than parallel ones. It is known that in the asymptotic regime, adaptive strategies confer no advantage over non-adaptive ones in asymmetric channel discrimination \cite{Hayashi2009,Berta2018b,Fang2019}. However, in practical situations of interest with a finite number of uses of the unknown channel and specific distinguishability tasks, it is possible that adaptive strategies offer an advantage. 

The formulation of quantum strategies offers a powerful framework in which to analyze this problem. Consider a strategy $\Nn$ such as the one in Figure~\ref{fig:strategy-costrategy-figure} that consists of $n$ uses of the channel $\mathcal{N}_{A \rightarrow B}$. This strategy can be made to interact with a general $n$-turn co-strategy $\Sn$, which encapsulates all possible adaptive operations. To study parallel strategies, $\Nn$ can also be made to interact with a constrained, parallel $n$-turn co-strategy. These two cases are described in Figure~\ref{fig:parallel-adaptive-strategies}. 

In this work, we study the gap between adaptive and parallel strategies for distinguishability tasks involving two different generalized amplitude damping channels (GADCs) \cite{Nielsen2010}. The GADC is a qubit-to-qubit channel that is characterized by a damping parameter $\gamma \in [0,1]$ and a noise parameter $N \in [0,1]$. It models the dynamics of a qubit system that is in contact with a thermal bath. It is used to describe some of the noise in superconducting-circuit based quantum computers \cite{Chirolli2008}. We consider two strategies $\Nn$ and $\Mn$ that each consist of $n$ uses of a particular GADC. The Choi operator of a GADC $\mathcal{A}_{\gamma,N}$ with damping parameter $\gamma$ and noise parameter $N$ is given by
\begin{equation}
	\Gamma_{RB}^{\mathcal{A}_{\gamma,N}}  \coloneqq 
	\begin{bmatrix}
		1-\gamma N & 0 & 0 & \sqrt{1-\gamma}\\
		0 & \gamma N & 0 & 0\\
		0 & 0 & \gamma\left(  1-N\right)  & 0\\
		\sqrt{1-\gamma} & 0 & 0 & 1-\gamma\left(  1-N\right)
	\end{bmatrix}.
\end{equation}

In Figure \ref{fig:discrimination-difference-plot}, we plot the \textit{difference} between the strategy distance $\frac{1}{2} \left\Vert \Nn - \Mn \right\Vert_{\diamond n}$ and the diamond distance $\frac{1}{2} \left\Vert \mathcal{N}^{\otimes n} - \mathcal{M}^{\otimes n} \right\Vert_{\diamond}$. We consider two GADCs with damping parameter $\gamma=0.2$ and $0.3$ respectively, while varying their common noise parameter $N$. We note here that the strategy distance $\frac{1}{2} \Vert \Nn - \Mn \Vert_{\diamond n}$ involves an optimization over all co-strategies, whereas the optimization involved in the diamond distance $\frac{1}{2} \left\Vert\mathcal{N}^{\otimes n} - \mathcal{M}^{\otimes n} \right\Vert_{\diamond}$ is restricted to parallel co-strategies. This enables us to investigate if adaptive strategies offer an advantage over parallel ones in channel discrimination, and we indeed see in Figure~\ref{fig:discrimination-difference-plot} that there is a non-zero gap between the strategy distance and the diamond distance.

\begin{figure} \centering
\includegraphics[width=3.2in]{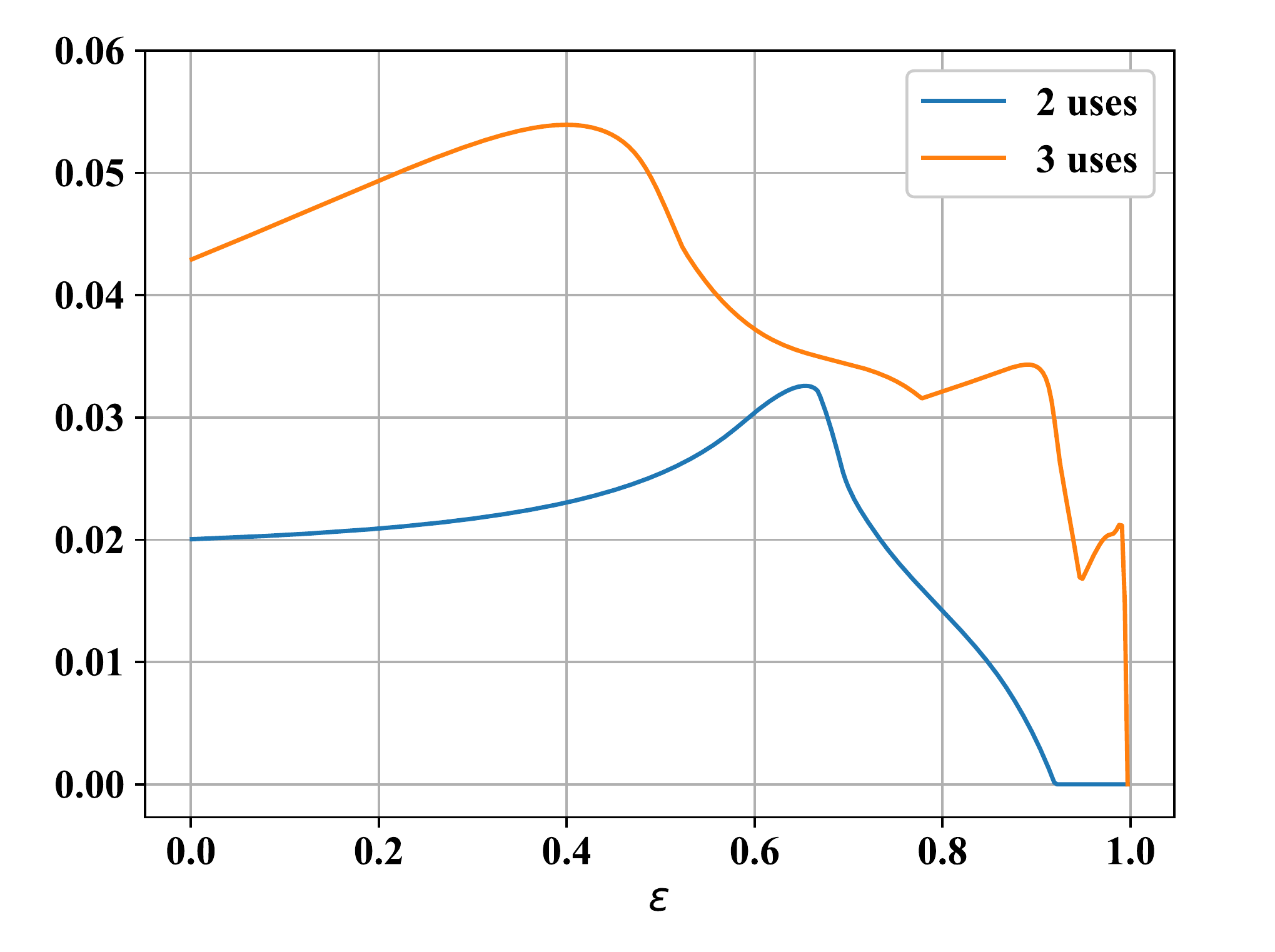}
\caption{Consider two GADCs with $N = 0.2$ and $0.3$ respectively. They both have $\gamma = 0.2$. We plot the \textit{difference} between the distillable distinguishabilities, which is given by $D_{\min}^{\varepsilon} ( \Nn \Vert \Mn )$ for the most general case and by $D_{\min}^{\varepsilon} ( \mathcal{N}^{\otimes n} \Vert \mathcal{M}^{\otimes n} )$ for the case when a distinguisher is limited to a parallel strategy.}
\label{fig:dmin-plot}
\end{figure}

\begin{figure} \centering
	\includegraphics[width=3.2in]{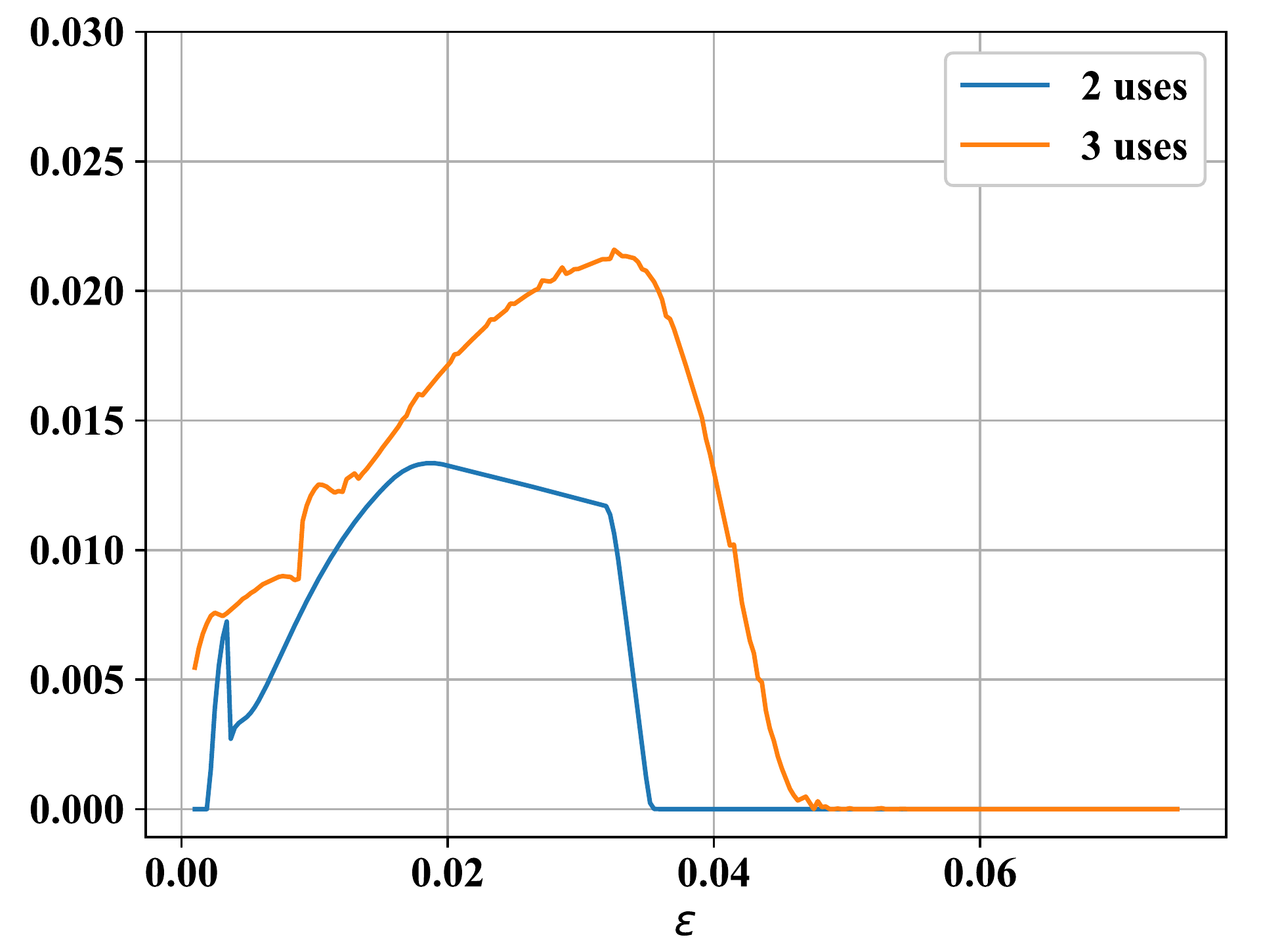}
	\caption{Consider two GADCs with $N = 0.2$ and $0.3$, respectively. They both have $\gamma = 0.2$. We plot the \textit{difference} between the approximate distinguishability cost, which is given by $D_{\max}^{\varepsilon} ( \Nn \Vert \Mn )$ for the most general case and by $D_{\max}^{\varepsilon} ( \mathcal{N}^{\otimes n} \Vert \mathcal{M}^{\otimes n} )$ for the case when an adversary is limited to a parallel strategy.}
	\label{fig:dmax-plot}
\end{figure}

Further, for two GADCs, we study the gap between adaptive and parallel co-strategies for the distillable distinguishability, as well as the distinguishability cost. We consider that both channels have $\gamma=0.2$. The channels differ in their noise parameter $N$, with the first channel having $N=0.2$ and the second channel having $N=0.3$. We use the SDP formulations of the smooth strategy min-relative entropy in \eqref{eq:strategy-min-entropy-primal} and the smooth strategy max-relative entropy in \eqref{eq:strategy-max-entropy-primal} to perform this calculation. For the distillable distinguishability, the results of the calculation are in Figure \ref{fig:dmin-plot}, where we see that there is a gap in the distillable distinguishability between adaptive and parallel strategies. The results of the distinguishability cost are in Figure \ref{fig:dmax-plot}, where again we see a nonzero gap between adaptive and parallel strategies. The code used to produce these numerical results is available with the arXiv version of the paper. It is not yet clear to us how to explain this behavior qualitatively, and so we leave it for future work to do so. 

\section{Conclusion}

In summary, in this paper we reviewed and further developed the resource theory of asymmetric distinguishability for quantum strategies, which is a high-level and flexible framework with which to study quantum interactions. 
We provided semi-definite programs to calculate the distillable distinguishability and the distinguishability cost of quantum strategy boxes, which we used to compare the power of adaptive strategies to parallel ones. 
It is known that for channel discrimination and distillable distinguishability, parallel strategies are equally powerful as adaptive strategies in the asymptotic limit; however, an example we considered shows that adaptive strategies provide an advantage in general when one considers a finite number of channel uses. 

\textit{Note added}: After the first version of our paper was posted to the arXiv in January 2020, there has been further work conducted on this topic recently, with Refs.~\cite{Bavaresco2020, Salek2021} studying the difference between adaptive and parallel strategies in channel distinguishability tasks, and Refs.~\cite{Rexiti2021, Pereira2021} studying the task of discriminating between amplitude damping channels.

\paragraph*{Acknowledgements} VK acknowledges support from the LSU Economic Development Assistantship. MMW acknowledges support from the US National Science Foundation via Grant No. 1907615.

\bibliographystyle{unsrt}
\bibliography{strategies-paper-isit}

\newpage

\pagebreak

\appendix

\section{One-shot distillation and dilution of strategy boxes} \label{sec:one-shot-theorem-proof}

In the following, we provide the proof of Theorem~\ref{thm:approx-one-shot}, which is claimed in Ref.~\cite{Wang2019a}. For completeness, we restate the theorem below:

\begin{theorem}
	In Ref.~\cite{Wang2019a}, the approximate one-shot distillable distinguishability of the strategy box $\left(\Nn,\Mn\right)$ is equal to the smooth strategy min-relative entropy:
	\begin{equation}
	D_{d}^{\varepsilon}(\Nn, \Mn)=D_{\min}^{\varepsilon}(\Nn \Vert \Mn ), 
	\end{equation}
	and the approximate one-shot distinguishability cost is equal to the smooth strategy max-relative entropy:
	\begin{equation}
	D_{c}^{\varepsilon}(\Nn, \Mn )=D_{\max}^{\varepsilon}(\Nn \Vert \Mn). 
	\end{equation}
\end{theorem}

\subsection{One-shot exact distillable distinguishability is strategy min-relative entropy}

We first prove the inequality 
\begin{equation}
	D_d^{0} (\mathcal{N}^{(n)}, \mathcal{M}^{(n)} ) \geq D_{\min} (\mathcal{N}^{(n)} \Vert \mathcal{M}^{(n)}),
\end{equation}
Let $\Theta$ be an arbitrary $n$-turn co-strategy that interacts with strategies $\Nn$ or $\Mn$ to yield a state on $R_n B_n$. Consider the projector $0 \leq \Lambda_{R_n B_n} \leq I_{R_n B_n}$ onto the support of $\Theta \circ \Nn$. Consider a post-processing of the output state $\omega_{R_n B_n}$ as follows:
\begin{multline}
	\omega_{R_n B_n} \rightarrow \Tr \!\left[ \Lambda_{R_n B_n} \omega_{R_n B_n} \right] |0\rangle \!\langle 0 |_X \\ + \Tr \!\left[ \left( I_{R_n B_n} - \Lambda_{R_n B_n} \right) \omega_{R_n B_n} \right] |1\rangle \!\langle 1 |_X.
\end{multline}
If the unknown strategy is $\Nn$, then the interaction with $\Theta$ followed by the above post-processing  yields $|0\rangle \!\langle 0 |_X$. If the unknown strategy is $\Mn$, then the final state is $\pi_M$ with 
\begin{equation}
	M = \frac{1}{\Tr \!\left[ \Lambda_{R_n B_n} \left( \Theta \circ \Mn \right) \right]},
\end{equation}
or equivalently,
\begin{equation}
	\log_2 M = D_{\min} ( \Theta \circ \Nn \Vert \Theta \circ \Mn ).
\end{equation}
Taking a supremum over all interacting co-strategies $\Theta$, we get
\begin{align}
	D_d^0 ( \Nn, \Mn ) &\geq \sup_{\Theta} D_{\min} ( \Theta \circ \Nn \Vert \Theta \circ \Mn ) \\
	&= D_{\min} ( \Nn \Vert \Mn ). \label{eq:exact-dist-achievability}
\end{align}

Next we prove the opposite inequality 
\begin{equation}
	D_d^{0} (\mathcal{N}^{(n)}, \mathcal{M}^{(n)} ) \leq D_{\min} (\mathcal{N}^{(n)} \Vert \mathcal{M}^{(n)})	
\end{equation}
which is a consequence of the data-processing inequality for the $D_{\min}$ strategy divergence \cite{Wang2019a}.
Consider an arbitrary $n$-turn co-strategy $\Theta$ that interacts with $\Nn$ to give $|0\rangle \!\langle 0 |$, and with $\Mn$ to give $\pi_M$. Then we can write 
\begin{align}
	D_{\min}(\Nn \Vert \Mn) &\geq D_{\min} (\Theta \circ \Nn \Vert \Theta \circ \Mn) \\
	&= D_{\min} ( |0\rangle \!\langle 0 | \Vert \pi_M) \\
	&= \log_2 M,
\end{align}
which yields 
\begin{equation}
	D_{\min} (\mathcal{N}^{(n)} \Vert \mathcal{M}^{(n)}) \geq D_d^{0} (\mathcal{N}^{(n)}, \mathcal{M}^{(n)} ).	\label{eq:exact-dist-converse}
\end{equation}
Putting \eqref{eq:exact-dist-achievability} and \eqref{eq:exact-dist-converse} together, we get
\begin{equation}
	D_d^{0} (\mathcal{N}^{(n)}, \mathcal{M}^{(n)} ) = D_{\min} (\mathcal{N}^{(n)} \Vert \mathcal{M}^{(n)}). 
\end{equation}

\subsection{One-shot approximate distillable distinguishability is smooth strategy min-relative entropy}

Here our aim is to prove
\begin{equation}
	D_d^{\varepsilon}(\Nn, \Mn) = D_{\min}^{\varepsilon} (\Nn \Vert \Mn).
\end{equation}
First we prove the inequality
\begin{equation}
	D_d^{\varepsilon}(\Nn, \Mn) \geq D_{\min}^{\varepsilon} (\Nn \Vert \Mn)	.
\end{equation}
Let $\Theta$ be an arbitrary interacting $n$-turn co-strategy and $\Lambda_{R_n B_n}$ a corresponding measurement operator satisfying $0 \leq \Lambda_{R_n B_n} \leq I_{R_n B_n}$ and 
\begin{equation}
	\Tr [ \Lambda_{R_n B_n} ( \Theta \circ \Nn  ) ] \geq 1 - \varepsilon. \label{eq:epsilon-close-probability}
\end{equation}
Consider, as in the exact case, a post-processing of the final state $\omega_{R_n B_n}$ by the measurement channel $\mathcal{L}_{R_nR_n\to X}$:
\begin{multline}
	\mathcal{L}_{R_nR_n\to X}(\omega_{RB^n}) := \Tr \!\left[ \Lambda_{R_n B_n} \omega_{R_n B_n} \right] |0\rangle \!\langle 0 |_X  \\ + \Tr \!\left[ \left( I_{R_n B_n} - \Lambda_{R_n B_n} \right) \omega_{R_n B_n} \right] |1\rangle \!\langle 1 |_X. 
\end{multline}
Using \eqref{eq:epsilon-close-probability}, we can conclude that $\mathcal{L} \circ \Theta \circ \Nn \approx_{\varepsilon} |0\rangle \!\langle 0 |$. Further, for 
\begin{equation}
	M = \frac{1}{\Tr [ \Lambda_{R_n B_n} ( \Theta \circ \Mn )  ] },
\end{equation}
we have $\mathcal{L} \circ \Theta \circ \Mn = \pi_M$.
Taking a supremum over all interacting co-strategies $\Theta$ and measurement channels $\mathcal{L}_{R_nR_n\to X}$, we get
\begin{align}
	D_d^{\varepsilon} (\Nn, \Mn) &\geq \sup_{\Theta} D_{\min}^{\varepsilon} ( \Theta \circ \Nn  \Vert \Theta \circ \Mn  ) \notag
	\\
	&= D_{\min}^{\varepsilon} ( \Nn \Vert \Mn ). \label{eq:approx-dist-achievability}
\end{align}

Next, we use data processing to prove the reverse inequality
\begin{equation}
	D_d^{\varepsilon}(\Nn, \Mn) \leq D_{\min}^{\varepsilon} (\Nn \Vert \Mn).
\end{equation}
Consider an $n$-turn co-strategy $\Theta$ and measurement channel $\mathcal{L}_{R_nR_n\to X}$ such that $\frac{1}{2} \left\Vert \mathcal{L}\circ \Theta \circ \Nn - |0\rangle \!\langle 0 | \right\Vert_1 \leq \varepsilon$. By a direct calculation with trace distance, we find that
\begin{align}
	\varepsilon &\geq \frac{1}{2} \left\Vert \mathcal{L} \circ \Theta \circ \Nn - |0\rangle \!\langle 0 | \right\Vert_1 \\
	&= 1 - \Tr[\Lambda (\Theta \circ \Nn) ].
\end{align}
We conclude that $\Tr[\Lambda (\Theta \circ \Nn) ] \geq 1 - \varepsilon$. In the definition of $D_{\min}^{\varepsilon}( \Theta \circ \Nn \Vert \pi_M )$, we can take the final measurement operator to be $\Lambda_{RB}$. This leaves us with $\Tr \!\left[ \Lambda_{R_n B_n} (\Theta \circ \Nn)  \right] \geq 1 - \varepsilon$ and $\Tr \!\left[ \Lambda_{R_n B_n} (\Theta \circ \Mn)  \right] = 1/M$. Since the definition of $D_{\min}^{\varepsilon}$ for strategies involves an optimization over co-strategies and measurement operators, we conclude that
\begin{align}
	D_{\min}^{\varepsilon} ( \Nn \Vert \Mn ) &\geq D_{\min}^{\varepsilon} ( \mathcal{L} \circ \Theta \circ \Nn \Vert \pi_M ) \\														&\geq \log_2 M 	,	
\end{align}
where the last inequality follows from Ref.~\cite[Appendix~F-1]{Wang2019b}.
Since the scheme considered for distillation is arbitrary, we conclude that
\begin{equation}
	D_{\min}^{\varepsilon} ( \Nn \Vert \Mn ) \geq  D_d^{\varepsilon} ( \Nn, \Mn ). 
	\label{eq:approx-dist-converse}
\end{equation}
Combining \eqref{eq:approx-dist-achievability} and \eqref{eq:approx-dist-converse}, we obtain the desired result:
\begin{equation}
	D_d^{\varepsilon} ( \Nn, \Mn ) = D_{\min}^{\varepsilon} ( \Nn \Vert \Mn ).
\end{equation}

\subsection{One-shot exact distinguishability cost is strategy max-relative entropy}

First, we aim to prove the inequality
\begin{equation}
	D_c^{0} (\Nn, \Mn) \leq D_{\max}(\Nn \Vert \Mn).
\end{equation}
To do so, we first let $\lambda$ be such that 
\begin{equation} \label{eq:exact-strategy-comparison}
	\Nn \leq 2^{\lambda} \Mn. 
\end{equation}
This means that
\begin{equation}
	\mathcal{N}'^{(n)} := \frac{2^{\lambda} \Mn - \Nn }{2^{\lambda} - 1}
\end{equation}
is a quantum strategy. Further, if the Choi operators of $\Nn$ and $\Mn$ are $\Gamma^{\Nn}$ and $\Gamma^{\Mn}$ respectively, then
\begin{equation}
	\frac{2^{\lambda} \Gamma^{\Mn} - \Gamma^{\Nn}}{2^{\lambda} - 1}
\end{equation}
is the Choi operator of $\mathcal{N}'^{(n)}$ (by linearity).

Consider an arbitrary $n$-turn co-strategy that is made to act as follows, beginning with system $X$. It acts as follows:
\begin{equation}
	\sigma_X \rightarrow ( \Theta \circ \Nn ) \bra{0} \sigma_X \ket{0} + ( \Theta \circ \mathcal{N}'^{(n)} ) \bra{1} \sigma_X \ket{1}.
\end{equation}
In the case that $\sigma_X = |0\rangle \!\langle 0 |_X$, then the output is $\Theta \circ \Nn$. If the input is $\pi_M$ where $M = 2^{\lambda}$, then the output is $\Theta \circ \Mn$. 

For this particular choice of transformation, we obtain a distinguishability cost of $\lambda$, so if one optimizes over all protocols, one obtains $D_c^{0} \left(\Nn, \Mn\right) \leq \lambda$. Now if we optimize over all $\lambda$ such that \eqref{eq:exact-strategy-comparison} holds, we obtain
\begin{equation}
	D_c^{0} ( \Nn, \Mn ) \leq D_{\max} ( \Nn \Vert \Mn ). \label{eq:exact-cost-converse}
\end{equation}

The opposite inequality follows from the data processing inequality for the strategy max-relative entropy \cite{Wang2019a}. Let $\Theta$ be a strategy satisfying 
\begin{align}
	\Theta\left(|0\rangle \!\langle 0 |\right) &= \Nn, \text{ and } \\
	\Theta\left(\pi_M \right) &= \Mn,
\end{align}
with $\log_2 M = D_c^0  (\Nn , \Mn )$.
Then consider the following chain of reasoning:
\begin{align}
	\log_2 M &= D_c^0  (\Nn , \Mn ) \\
	&= D_{\max} ( |0\rangle \!\langle 0 | \Vert \pi_M ) \\
	&\geq D_{\max} \left( \Theta(|0\rangle \!\langle 0 |) \Vert \Theta(\pi_M) \right) \\
	&= D_{\max} ( \Nn \Vert \Mn ).
\end{align}
This lets us conclude that 
\begin{equation}
	D_c^{0} ( \Nn, \Mn ) \geq D_{\max} ( \Nn \Vert \Mn ). \label{eq:exact-cost-achievability}
\end{equation}
Putting together \eqref{eq:exact-cost-converse} and \eqref{eq:exact-cost-achievability}, we obtain the desired result, which is
\begin{equation}
	D_c^{0} ( \Nn, \Mn ) = D_{\max} ( \Nn \Vert \Mn ).
\end{equation}

\subsection{One-shot approximate distinguishability cost is smooth strategy max-relative entropy}

Here we aim to prove that
\begin{equation}
	D_c^{\varepsilon} ( \Nn, \Mn ) = D_{\max}^{\varepsilon} ( \Nn \Vert \Mn ).
\end{equation}
First, we prove the inequality
\begin{equation}
	D_c^{\varepsilon} ( \Nn, \Mn ) \leq D_{\max}^{\varepsilon} ( \Nn \Vert \Mn ).  
\end{equation}
To do so, we consider a quantum strategy $\mathcal{N}'^{(n)} \approx_{\varepsilon} \Nn$ (which means that $\frac{1}{2} \left\Vert \mathcal{N}'^{(n)} - \Nn \right\Vert \leq \varepsilon$). We use the construction for the exact distinguishability cost, but instead for $\mathcal{N}'^{(n)}$, and therefore obtain
\begin{equation}
	D_c^{\varepsilon} ( \Nn, \Mn ) \leq D_{\max}^{\varepsilon} ( \mathcal{N}'^{(n)} \Vert \Mn ).
\end{equation}
By optimizing the above over all $\mathcal{N}'^{(n)}$ satisfying $\mathcal{N}'^{(n)} \approx_{\varepsilon} \Nn$, we obtain
\begin{equation}
	D_c^{\varepsilon} ( \Nn, \Mn ) \leq D_{\max}^{\varepsilon} ( \Nn \Vert \Mn ). \label{eq:approx-cost-converse}
\end{equation}

To prove the reverse inequality 
\begin{equation}
	D_c^{\varepsilon} ( \Nn, \Mn ) \geq D_{\max}^{\varepsilon} ( \Nn \Vert \Mn ),
\end{equation}
we again use data-processing arguments. Consider first a strategy $\Theta$ such that 
\begin{align}
	\Theta \left( |0\rangle \!\langle 0 | \right) & \approx_{\varepsilon} \Nn, \text{ and} \\
	\Theta \left( \pi_M \right) &= \Mn,
\end{align}
with $\log_2 M = D_c^0  (\Nn , \Mn )$.
Now consider the following:
\begin{align} 
	D_c^{\varepsilon} ( \Nn, \Mn ) &= \log_2 M \\
	&= D_{\max} ( |0\rangle \!\langle 0 | \Vert \pi_M ) \\
	&\geq D_{\max} ( \Theta ( |0\rangle \!\langle 0 | ) \Vert \Theta ( \pi_M ) ) \\
	&= D_{\max} ( \Theta ( |0\rangle \!\langle 0 | ) \Vert \Mn ) \\
	&\geq D_{\max}^{\varepsilon} ( \Nn \Vert \Mn ) \label{eq:approx-cost-achievability}.
\end{align}
Putting together \eqref{eq:approx-cost-converse} and \eqref{eq:approx-cost-achievability}, we get the desired result.

\newpage
\pagebreak
\onecolumngrid

\section{Derivation of semi-definite programming duals}

Here, we provide full details of how to arrive at the duals for the semi-definite programs for the normalized strategy distance, the smooth strategy min-relative entropy, and the smooth strategy max-relative entropy.

\subsection{Background}

Suppose that a semi-definite program is given in primal form as follows:%
\begin{equation}
	\sup_{X\geq0}\left\{  \operatorname{Tr}[AX]:\Phi_{1}(X)=B_{1},\Phi_{2}(X)\leq
	B_{2}\right\}  , \label{eq:generic-SDP-primal}%
\end{equation}
then its dual is given by%
\begin{equation}
	\inf_{Y_{1}\in\text{Herm},Y_{2}\geq0}\big\{  \operatorname{Tr}[B_{1}%
	Y_{1}]+\operatorname{Tr}[B_{2}Y_{2}]: \Phi_{1}^{\dag}(Y_{1})+\Phi_{2}^{\dag
	}(Y_{2})\geq A\big\}  . \label{eq:generic-SDP-dual}%
\end{equation}
We use
\eqref{eq:generic-SDP-primal} and \eqref{eq:generic-SDP-dual}\ in the
forthcoming sections to derive the various duals presented in our paper.

Alternatively, the following is useful as well. If the primal can be written
as%
\begin{equation}
	\inf_{X\geq0}\left\{  \operatorname{Tr}[(-A)X]:\Phi_{1}(X)=B_{1},\Phi
	_{2}(X)\leq B_{2}\right\}  . \label{eq:other-opt-primal}%
\end{equation}
Then its dual is given by%
\begin{equation}
	\sup_{Y_{1}\in\text{Herm},Y_{2}\geq0}\big\{  -\operatorname{Tr}[B_{1}%
	Y_{1}]-\operatorname{Tr}[B_{2}Y_{2}]: \Phi_{1}^{\dag}(Y_{1})+\Phi_{2}^{\dag
	}(Y_{2})\geq A\big\}  , \label{eq:other-opt-dual}%
\end{equation}
which comes about by applying a minus sign to \eqref{eq:generic-SDP-primal}
and carrying it through.

\subsection{Normalized strategy distance}

First, we repeat the primal for the normalized strategy distance given in~\eqref{eq:primal-strategy-dist} in the main text:%
\begin{equation}
	\sup_{S,S_{[1]},\ldots,S_{[n]}\geq0}\operatorname{Tr}[S(\Gamma^{\mathcal{N}%
		^{(n)}}-\Gamma^{\mathcal{M}^{(n)}})]
\end{equation}
subject to%
\begin{align}
	S &  \leq S_{[n]}\otimes I_{B_{n}},\\
	\operatorname{Tr}_{A_{n}}[S_{[n]}] &  =S_{[n-1]}\otimes I_{B_{n-1}},\\
	\operatorname{Tr}_{A_{n-1}}[S_{[n-1]}] &  =S_{[n-2]}\otimes I_{B_{n-2}},\\
	&  \vdots\notag\\
	\operatorname{Tr}_{A_{2}}[S_{[2]}] &  =S_{[1]}\otimes I_{B_{1}},\\
	\operatorname{Tr}[S_{[1]}] &  =1.
\end{align}
Now mapping to \eqref{eq:generic-SDP-primal}, we find that%
\begin{align}
	X &  =\text{diag}(S,S_{[n]},S_{[n-1]},\ldots,S_{[2]},S_{[1]}),\\
	A &  =\text{diag}(\Gamma^{\mathcal{N}^{(n)}}-\Gamma^{\mathcal{M}^{(n)}%
	},0,0,\ldots,0,0),\\
	\Phi_{1}(X) &  =\text{diag}(\operatorname{Tr}_{A_{n}}[S_{[n]}]-S_{[n-1]}%
	\otimes I_{B_{n-1}},\nonumber\\
	&  \qquad\operatorname{Tr}_{A_{n-1}}[S_{[n-1]}]-S_{[n-2]}\otimes I_{B_{n-2}%
	},\ldots,\nonumber\\
	&  \qquad\operatorname{Tr}_{A_{2}}[S_{[2]}]-S_{[1]}\otimes I_{B_{1}%
	},\operatorname{Tr}[S_{[1]}]),\\
	B_{1} &  =\text{diag}(0,0,\ldots,0,1),\\
	\Phi_{2}(X) &  =S-S_{[n]}\otimes I_{B_{n}},\\
	B_{2} &  =0.
\end{align}
We should now figure out the adjoints of $\Phi_{1}$ and $\Phi_{2}$. Consider
that%
\begin{equation}
	\operatorname{Tr}[Y_{i}\Phi_{i}(X)]=\operatorname{Tr}[\Phi_{i}^{\dag}%
	(Y_{i})X]\text{ for }i\in\left\{  1,2\right\}  .
\end{equation}
Set%
\begin{equation}
	Y_{1}=\text{diag}(Z_{n-1},Z_{n-2},\ldots,Z_{1},\mu).\label{eq:Y1-var}%
\end{equation}
Then we find that%
\begin{align}
	&  \operatorname{Tr}[Y_{1}\Phi_{1}(X)]\nonumber\\
	&  =\operatorname{Tr}[Z_{n-1}\left(  \operatorname{Tr}_{A_{n}}[S_{[n]}%
	]-S_{[n-1]}\otimes I_{B_{n-1}}\right)  ]\nonumber\\
	&  \qquad+\operatorname{Tr}[Z_{n-2}(\operatorname{Tr}_{A_{n-1}}[S_{[n-1]}%
	]-S_{[n-2]}\otimes I_{B_{n-2}})]+\cdots\nonumber\\
	&  \qquad+\operatorname{Tr}[Z_{1}(\operatorname{Tr}_{A_{2}}[S_{[2]}%
	]-S_{[1]}\otimes I_{B_{1}})]+\mu\operatorname{Tr}[S_{[1]}]\\
	&  =\operatorname{Tr}[(Z_{n-1}\otimes I_{A_{n}})S_{[n]}]-\operatorname{Tr}%
	[\operatorname{Tr}_{B_{n-1}}[Z_{n-1}]S_{[n-1]}]\nonumber\\
	&  \qquad+\operatorname{Tr}[(Z_{n-2}\otimes I_{A_{n-1}})S_{[n-1]}%
	]-\operatorname{Tr}[\operatorname{Tr}_{B_{n-2}}[Z_{n-2}]S_{[n-2]}%
	]+\cdots\nonumber\\
	&  \qquad+\operatorname{Tr}[(Z_{1}\otimes I_{A_{2}})S_{[2]}]-\operatorname{Tr}%
	[\operatorname{Tr}_{B_{1}}[Z_{1}]S_{[1]}]+\operatorname{Tr}[\mu I_{A_{1}%
	}S_{[1]}]\\
	&  =\operatorname{Tr}[(Z_{n-1}\otimes I_{A_{n}})S_{[n]}]+\operatorname{Tr}%
	[(Z_{n-2}\otimes I_{A_{n-1}}-\operatorname{Tr}_{B_{n-1}}[Z_{n-1}%
	])S_{[n-1]}]\nonumber\\
	&  \qquad+\operatorname{Tr}[(Z_{n-3}\otimes I_{A_{n-2}}-\operatorname{Tr}%
	_{B_{n-2}}[Z_{n-2}])S_{[n-2]}]+\cdots\nonumber\\
	&  \qquad+\operatorname{Tr}[(Z_{1}\otimes I_{A_{2}}-\operatorname{Tr}_{B_{2}%
	}[Z_{2}])S_{[2]}]+\operatorname{Tr}[(\mu I_{A_{1}}-\operatorname{Tr}_{B_{1}%
	}[Z_{1}])S_{[1]}].
\end{align}
So this implies that%
\begin{equation}
	\Phi_{1}^{\dag}(Y_{1})=\text{diag}(0,Z_{n-1}\otimes I_{A_{n}},Z_{n-2}\otimes
	I_{A_{n-1}}-\operatorname{Tr}_{B_{n-1}}[Z_{n-1}], \ldots,Z_{1}\otimes I_{A_{2}}-\operatorname{Tr}_{B_{2}}[Z_{2}],\mu I_{A_{1}%
	}-\operatorname{Tr}_{B_{1}}[Z_{1}]). \label{eq:adjoint-phi-1}
\end{equation}
Now set%
\begin{equation}
	Y_{2}=Z_{n},
\end{equation}
and we find that%
\begin{align}
	\operatorname{Tr}[Y_{2}\Phi_{2}(X)] &  =\operatorname{Tr}[Z_{n}(S-S_{[n]}%
	\otimes I_{B_{n}})]\\
	&  =\operatorname{Tr}[Z_{n}S]-\operatorname{Tr}[\operatorname{Tr}_{B_{n}%
	}[Z_{n}]S_{[n]}],
\end{align}
so that%
\begin{equation}
	\Phi_{2}^{\dag}(Y_{2})=\text{diag}(Z_{n},-\operatorname{Tr}_{B_{n}}%
	[Z_{n}],0,\ldots,0,0).
\end{equation}
We finally find that%
\begin{multline}
	\Phi_{1}^{\dag}(Y_{1})+\Phi_{2}^{\dag}(Y_{2})=\text{diag}(Z_{n},Z_{n-1}\otimes
	I_{A_{n}}-\operatorname{Tr}_{B_{n}}[Z_{n}], Z_{n-2}\otimes I_{A_{n-1}}-\operatorname{Tr}_{B_{n-1}}[Z_{n-1}],\ldots,\\
	Z_{1}\otimes I_{A_{2}}-\operatorname{Tr}_{B_{2}}[Z_{2}],\mu I_{A_{1}%
	}-\operatorname{Tr}_{B_{1}}[Z_{1}]),
\end{multline}
and so $\Phi_{1}^{\dag}(Y_{1})+\Phi_{2}^{\dag}(Y_{2})\geq A$ is equivalent to
the following conditions:%
\begin{align}
	Z_{n} &  \geq\Gamma^{\mathcal{N}^{(n)}}-\Gamma^{\mathcal{M}^{(n)}},\\
	Z_{n-1}\otimes I_{A_{n}} &  \geq\operatorname{Tr}_{B_{n}}[Z_{n}],\\
	Z_{n-2}\otimes I_{A_{n-1}} &  \geq\operatorname{Tr}_{B_{n-1}}[Z_{n-1}],\\
	&  \vdots\notag\\
	Z_{1}\otimes I_{A_{2}} &  \geq\operatorname{Tr}_{B_{2}}[Z_{2}],\\
	\mu I_{A_{1}} &  \geq\operatorname{Tr}_{B_{1}}[Z_{1}].
\end{align}
So then the dual is given by plugging into \eqref{eq:generic-SDP-dual}:%
\begin{equation}
	\inf_{\substack{Z_{n-1},Z_{n-2},\ldots,Z_{1}\in\text{Herm},\\\mu\in
			\mathbb{R},Z_{n}\geq0}}\mu,
\end{equation}
subject to%
\begin{align}
	Z_{n}  & \geq\Gamma^{\mathcal{N}^{(n)}}-\Gamma^{\mathcal{M}^{(n)}},\\
	Z_{n-1}\otimes I_{A_{n}}  & \geq\operatorname{Tr}_{B_{n}}[Z_{n}],\\
	Z_{n-2}\otimes I_{A_{n-1}}  & \geq\operatorname{Tr}_{B_{n-1}}[Z_{n-1}],\\
	& \vdots \notag\\
	Z_{1}\otimes I_{A_{2}}  & \geq\operatorname{Tr}_{B_{2}}[Z_{2}],\\
	\mu I_{A_{1}}  & \geq\operatorname{Tr}_{B_{1}}[Z_{1}].
\end{align}

\subsection{Smooth strategy min-relative entropy}

First, we repeat the primal for the smooth strategy min-relative entropy given  in~\eqref{eq:strategy-min-entropy-primal} in the main text:%
\begin{equation}
	\inf_{S,S_{[1]},\ldots,S_{[n]}\geq0}\operatorname{Tr}[S\ \Gamma^{\mathcal{M}%
		^{(n)}}]
\end{equation}
subject to%
\begin{align}
	\operatorname{Tr}[S\ \Gamma^{\mathcal{N}^{(n)}}]  &  \geq1-\varepsilon\\
	S  &  \leq S_{[n]}\otimes I_{B_{n}},\\
	\operatorname{Tr}_{A_{n}}[S_{[n]}]  &  =S_{[n-1]}\otimes I_{B_{n-1}},\\
	\operatorname{Tr}_{A_{n-1}}[S_{[n-1]}]  &  =S_{[n-2]}\otimes I_{B_{n-2}},\\
	&  \vdots\notag\\
	\operatorname{Tr}_{A_{2}}[S_{[2]}]  &  =S_{[1]}\otimes I_{B_{1}},\\
	\operatorname{Tr}[S_{[1]}]  &  =1.
\end{align}
Now mapping to \eqref{eq:other-opt-primal}, we find that%
\begin{align}
	X  &  =\text{diag}(S,S_{[n]},S_{[n-1]},\ldots,S_{[2]},S_{[1]}),\\
	A  &  =\text{diag}(-\Gamma^{\mathcal{M}^{(n)}},0,0,\ldots,0,0),\\
	\Phi_{1}(X)  &  =\text{diag}(\operatorname{Tr}_{A_{n}}[S_{[n]}]-S_{[n-1]}%
	\otimes I_{B_{n-1}},\nonumber\\
	&  \qquad\operatorname{Tr}_{A_{n-1}}[S_{[n-1]}]-S_{[n-2]}\otimes I_{B_{n-2}%
	},\ldots,\nonumber\\
	&  \qquad\operatorname{Tr}_{A_{2}}[S_{[2]}]-S_{[1]}\otimes I_{B_{1}%
	},\operatorname{Tr}[S_{[1]}]),\\
	B_{1}  &  =\text{diag}(0,0,\ldots,0,1),\\
	\Phi_{2}(X)  &  =\text{diag}(-\operatorname{Tr}[S\ \Gamma^{\mathcal{N}^{(n)}%
	}],S-S_{[n]}\otimes I_{B_{n}}),\\
	B_{2}  &  =\text{diag}(-\left(  1-\varepsilon\right)  ,0).
\end{align}
The variable $Y_{1}$ is the same as in \eqref{eq:Y1-var}, and the adjoint of
$\Phi_{1}$ is the same as in \eqref{eq:adjoint-phi-1}. Let us set%
\begin{equation}
	Y_{2}=\text{diag}(\mu_{1},Z_{n}).
\end{equation}
Then we find that%
\begin{align}
	\operatorname{Tr}[Y_{2}\Phi_{2}(X)]  &  =-\mu_{1}\operatorname{Tr}%
	[S\ \Gamma^{\mathcal{N}^{(n)}}]+\operatorname{Tr}[Z_{n}(S-S_{[n]}\otimes
	I_{B_{n}})]\\
	&  =\operatorname{Tr}[(Z_{n}-\mu_{1}\Gamma^{\mathcal{N}^{(n)}}%
	)S]-\operatorname{Tr}[\operatorname{Tr}_{B_{n}}[Z_{n}]S_{[n]}],
\end{align}
so that%
\begin{equation}
	\Phi_{2}^{\dag}(Y_{2})=\text{diag}(Z_{n}-\mu_{1}\Gamma^{\mathcal{N}^{(n)}%
	},-\operatorname{Tr}_{B_{n}}[Z_{n}],0,\ldots,0,0).
\end{equation}
So then we find that%
\begin{multline}
	\Phi_{1}^{\dag}(Y_{1})+\Phi_{2}^{\dag}(Y_{2})=\text{diag}(Z_{n}-\mu_{1}%
	\Gamma^{\mathcal{N}^{(n)}},Z_{n-1}\otimes I_{A_{n}}-\operatorname{Tr}_{B_{n}%
	}[Z_{n}], Z_{n-2}\otimes I_{A_{n-1}}-\operatorname{Tr}_{B_{n-1}}[Z_{n-1}],\ldots,\\
	Z_{1}\otimes I_{A_{2}}-\operatorname{Tr}_{B_{2}}[Z_{2}],\mu I_{A_{1}%
	}-\operatorname{Tr}_{B_{1}}[Z_{1}])
\end{multline}
and so $\Phi_{1}^{\dag}(Y_{1})+\Phi_{2}^{\dag}(Y_{2})\geq A$ is equivalent to
the following conditions:%
\begin{align}
	Z_{n}-\mu_{1}\Gamma^{\mathcal{N}^{(n)}}  &  \geq-\Gamma^{\mathcal{M}^{(n)}},\\
	Z_{n-1}\otimes I_{A_{n}}  &  \geq\operatorname{Tr}_{B_{n}}[Z_{n}],\\
	Z_{n-2}\otimes I_{A_{n-1}}  &  \geq\operatorname{Tr}_{B_{n-1}}[Z_{n-1}],\\
	&  \vdots\notag\\
	Z_{1}\otimes I_{A_{2}}  &  \geq\operatorname{Tr}_{B_{2}}[Z_{2}],\\
	\mu I_{A_{1}}  &  \geq\operatorname{Tr}_{B_{1}}[Z_{1}].
\end{align}
Also, observe that%
\begin{equation}
	\operatorname{Tr}[B_{1}Y_{1}]+\operatorname{Tr}[B_{2}Y_{2}]=\mu-\mu_{1}\left(
	1-\varepsilon\right)
\end{equation}
Thus, we conclude after plugging into \eqref{eq:other-opt-dual}\ that the dual
is given by%
\begin{equation}
	\sup_{\substack{\mu_{1},Z_{n}\geq0,\\\mu\in\mathbb{R},Z_{1},\ldots,Z_{n-1}%
			\in\text{Herm}}}\mu_{1}\left(  1-\varepsilon\right)  -\mu
\end{equation}
subject to%
\begin{align}
	Z_{n}  &  \geq\mu_{1}\Gamma^{\mathcal{N}^{(n)}}-\Gamma^{\mathcal{M}^{(n)}},\\
	Z_{n-1}\otimes I_{A_{n}}  &  \geq\operatorname{Tr}_{B_{n}}[Z_{n}],\\
	Z_{n-2}\otimes I_{A_{n-1}}  &  \geq\operatorname{Tr}_{B_{n-1}}[Z_{n-1}],\\
	&  \vdots\notag\\
	Z_{1}\otimes I_{A_{2}}  &  \geq\operatorname{Tr}_{B_{2}}[Z_{2}],\\
	\mu I_{A_{1}}  &  \geq\operatorname{Tr}_{B_{1}}[Z_{1}].
\end{align}

\subsection{Smooth strategy max-relative entropy}

First, we repeat the primal form of the smooth strategy max-relative entropy, given in~\eqref{eq:strategy-max-entropy-primal} in the main text:
\begin{equation}
	\inf_{\substack{\lambda,Y_{n},N\geq0,\\N_{[n-1]},\ldots,N_{[1]}\geq
			0,\\Y_{1},\ldots,Y_{n-1}\in\text{Herm}}}\lambda
\end{equation}
subject to%
\begin{align}
	N  &  \leq\lambda\Gamma^{\mathcal{M}^{(n)}},\\
	Y_{n}  &  \geq\Gamma^{\mathcal{N}^{(n)}}-N,\\
	Y_{n-1}\otimes I_{A_{n}}  &  \geq\operatorname{Tr}_{B_{n}}[Y_{n}],\\
	Y_{n-2}\otimes I_{A_{n-1}}  &  \geq\operatorname{Tr}_{B_{n-1}}[Y_{n-1}],\\
	&  \vdots\notag\\
	Y_{1}\otimes I_{A_{2}}  &  \geq\operatorname{Tr}_{B_{2}}[Y_{2}],\\
	\varepsilon I_{A_{1}}  &  \geq\operatorname{Tr}_{B_{1}}[Y_{1}],\\
	\operatorname{Tr}_{B_{n}}[N]  &  =N_{[n-1]}\otimes I_{A_{n}},\\
	\operatorname{Tr}_{B_{n-1}}[N_{[n-1]}]  &  =N_{[n-2]}\otimes I_{A_{n-1}},\\
	&  \vdots\notag\\
	\operatorname{Tr}_{B_{2}}[N_{[2]}]  &  =N_{[1]}\otimes I_{A_{2}},\\
	\operatorname{Tr}_{B_{1}}[N_{[1]}]  &  =I_{A_{1}}.
\end{align}
As a consequence of $Y_{n}\geq0$ and the constraints above, it follows that
$Y_{n-1},\ldots,Y_{1}\geq0$. So the above SDP\ can be cast in the form of
\eqref{eq:other-opt-primal}, with%
\begin{align}
	X  &  =\left(  \lambda,Y_{n},Y_{n-1},\ldots,Y_{2},Y_{1},N,N_{[n-1]}%
	,\ldots,N_{[2]},N_{[1]}\right)  ,\\
	A  &  =(-1,0,0,\ldots,0,0,0,\ldots,0),\\
	\Phi_{1}(X)  &  =\text{diag}(\operatorname{Tr}_{B_{n}}[N]-N_{[n-1]}\otimes
	I_{A_{n}}, \operatorname{Tr}_{B_{n-1}}[N_{[n-1]}]-N_{[n-2]}\otimes I_{A_{n-1}%
	},\ldots,\nonumber\\
	&  \qquad\operatorname{Tr}_{B_{2}}[N_{[2]}]-N_{[1]}\otimes I_{A_{2}%
	},\operatorname{Tr}_{B_{1}}[N_{[1]}]),\\
	B_{1}  &  =\text{diag}(0,0,\ldots,0,I_{A_{1}}),\\
	\Phi_{2}(X)  &  =\text{diag}(N-\lambda\Gamma^{\mathcal{M}^{(n)}}%
	,-Y_{n}-N,\operatorname{Tr}_{B_{n}}[Y_{n}]-Y_{n-1}\otimes I_{A_{n}%
	},\nonumber\\
	&  \qquad\operatorname{Tr}_{B_{n-1}}[Y_{n-1}]-Y_{n-2}\otimes I_{A_{n-1}%
	},\ldots, \operatorname{Tr}_{B_{2}}[Y_{2}]-Y_{1}\otimes I_{A_{2}}%
	,\operatorname{Tr}_{B_{1}}[Y_{1}]),\\
	B_{2}  &  =\text{diag}(0,-\Gamma^{\mathcal{N}^{(n)}},0,\ldots,0,\varepsilon
	I_{A_{1}}).
\end{align}
We should now figure out the adjoints of $\Phi_{1}$ and $\Phi_{2}$. Consider
that%
\begin{equation}
	\operatorname{Tr}[Y_{i}\Phi_{i}(X)]=\operatorname{Tr}[\Phi_{i}^{\dag}%
	(Y_{i})X]\text{ for }i\in\left\{  1,2\right\}  .
\end{equation}
Set%
\begin{equation}
	Y_{1}=\text{diag}(Z_{n},Z_{n-1},\ldots,Z_{2},Z_{1}).
\end{equation}
Then consider that%
\begin{align}
	&  \operatorname{Tr}[Y_{1}\Phi_{1}(X)]\nonumber\\
	&  =\operatorname{Tr}[Z_{n}(\operatorname{Tr}_{B_{n}}[N]-N_{[n-1]}\otimes
	I_{A_{n}})]\nonumber\\
	&  \qquad+\operatorname{Tr}[Z_{n-1}(\operatorname{Tr}_{B_{n-1}}[N_{[n-1]}%
	]-N_{[n-2]}\otimes I_{A_{n-1}})]+\cdots\nonumber\\
	&  \qquad+\operatorname{Tr}[Z_{2}(\operatorname{Tr}_{B_{2}}[N_{[2]}%
	]-N_{[1]}\otimes I_{A_{2}})]+\operatorname{Tr}[Z_{1}\operatorname{Tr}_{B_{1}%
	}[N_{[1]}]]\\
	&  =\operatorname{Tr}[(Z_{n}\otimes I_{B_{n}})N]-\operatorname{Tr}%
	[\operatorname{Tr}_{A_{n}}[Z_{n}]N_{[n-1]}]\nonumber\\
	&  \qquad+\operatorname{Tr}[(Z_{n-1}\otimes I_{B_{n-1}})N_{[n-1]}%
	]-\operatorname{Tr}[\operatorname{Tr}_{A_{n-1}}[Z_{n-1}]N_{[n-2]}%
	]+\cdots\nonumber\\
	&  \qquad+\operatorname{Tr}[(Z_{2}\otimes I_{B_{2}})N_{[2]}]-\operatorname{Tr}%
	[\operatorname{Tr}_{A_{2}}[Z_{2}]N_{[1]}]+\operatorname{Tr}[(Z_{1}\otimes
	I_{B_{1}})N_{[1]}]\\
	&  =\operatorname{Tr}[(Z_{n}\otimes I_{B_{n}})N]+\operatorname{Tr}%
	[(Z_{n-1}\otimes I_{B_{n-1}}-\operatorname{Tr}_{A_{n}}[Z_{n}])N_{[n-1]}%
	]\nonumber\\
	&  \qquad+\operatorname{Tr}[(Z_{n-2}\otimes I_{B_{n-2}}-\operatorname{Tr}%
	_{A_{n-1}}[Z_{n-1}])N_{[n-2]}]+\cdots\nonumber\\
	&  \qquad+\operatorname{Tr}[(Z_{2}\otimes I_{B_{2}}-\operatorname{Tr}_{A_{3}%
	}[Z_{3}])N_{[2]}] + \operatorname{Tr}[(Z_{1}\otimes I_{B_{1}}-\operatorname{Tr}_{A_{2}%
}[Z_{2}])N_{[1]}].
\end{align}
So then the adjoint of $\Phi_{1}$ is given by%
\begin{multline}
	\Phi_{1}^{\dag}(Y_{1})=\text{diag}(0,0,0,\ldots,0,0,Z_{n}\otimes I_{B_{n}%
	},Z_{n-1}\otimes I_{B_{n-1}}-\operatorname{Tr}_{A_{n}}[Z_{n}],\ldots,\\
	Z_{2}\otimes I_{B_{2}}-\operatorname{Tr}_{A_{3}}[Z_{3}],Z_{1}\otimes I_{B_{1}%
	}-\operatorname{Tr}_{A_{2}}[Z_{2}]).
\end{multline}
Set%
\begin{equation}
	Y_{2}=\text{diag}(W_{n+2},W_{n+1},W_{n},W_{n-1},\ldots,W_{2},W_{1}).
\end{equation}
Then we find that%
\begin{align}
	&  \operatorname{Tr}[Y_{2}\Phi_{2}(X)]\nonumber\\
	&  =\operatorname{Tr}[W_{n+2}(N-\lambda\Gamma^{\mathcal{M}^{(n)}%
	})]-\operatorname{Tr}[W_{n+1}(Y_{n}+N)]\nonumber\\
	&  \qquad+\operatorname{Tr}[W_{n}(\operatorname{Tr}_{B_{n}}[Y_{n}%
	]-Y_{n-1}\otimes I_{A_{n}})]\nonumber\\
	&  \qquad+\operatorname{Tr}[W_{n-1}(\operatorname{Tr}_{B_{n-1}}[Y_{n-1}%
	]-Y_{n-2}\otimes I_{A_{n-1}})]+\cdots\nonumber\\
	&  \qquad+\operatorname{Tr}[W_{2}(\operatorname{Tr}_{B_{2}}[Y_{2}%
	]-Y_{1}\otimes I_{A_{2}})]+\operatorname{Tr}[W_{1}\operatorname{Tr}_{B_{1}%
	}[Y_{1}]]\\
	&  =\operatorname{Tr}[W_{n+2}N]-\lambda\operatorname{Tr}[W_{n+2}%
	\Gamma^{\mathcal{M}^{(n)}}]-\operatorname{Tr}[W_{n+1}(Y_{n}+N)]\nonumber\\
	&  \qquad+\operatorname{Tr}[(W_{n}\otimes I_{B_{n}})Y_{n}]-\operatorname{Tr}%
	[\operatorname{Tr}_{A_{n}}[W_{n}]Y_{n-1}]\nonumber\\
	&  \qquad+\operatorname{Tr}[(W_{n-1}\otimes I_{B_{n-1}})Y_{n-1}%
	]-\operatorname{Tr}[\operatorname{Tr}_{A_{n-1}}[W_{n-1}]Y_{n-2}]+\cdots
	\nonumber\\
	&  \qquad+\operatorname{Tr}[(W_{2}\otimes I_{B_{2}})Y_{2}]-\operatorname{Tr}%
	[\operatorname{Tr}_{A_{2}}[W_{2}]Y_{1}]+\operatorname{Tr}[(W_{1}\otimes
	I_{B_{1}})Y_{1}]\\
	&  =\operatorname{Tr}[(W_{n+2}-W_{n+1})N]-\lambda\operatorname{Tr}%
	[W_{n+2}\Gamma^{\mathcal{M}^{(n)}}]+\operatorname{Tr}[(W_{n}\otimes I_{B_{n}%
	}-W_{n+1})Y_{n}]\nonumber\\
	&  \qquad+\operatorname{Tr}[(W_{n-1}\otimes I_{B_{n-1}}-\operatorname{Tr}%
	_{A_{n}}[W_{n}])Y_{n-1}]\nonumber\\
	&  \qquad+\operatorname{Tr}[(W_{n-2}\otimes I_{B_{n-2}}-\operatorname{Tr}%
	_{A_{n-1}}[W_{n-1}])Y_{n-2}]+\cdots\nonumber\\
	&  \qquad+\operatorname{Tr}[(W_{2}\otimes I_{B_{2}}-\operatorname{Tr}_{A_{3}%
	}[W_{3}])Y_{2}]+\operatorname{Tr}[(W_{1}\otimes I_{B_{1}}-\operatorname{Tr}%
	_{A_{2}}[W_{2}])Y_{1}].
\end{align}

So then the adjoint of $\Phi_{2}$ is given by%
\begin{multline}
	\Phi_{2}^{\dag}(Y_{2})=\text{diag}(-\operatorname{Tr}[W_{n+2}\Gamma
	^{\mathcal{M}^{(n)}}],W_{n}\otimes I_{B_{n}}-W_{n+1}, W_{n-1}\otimes I_{B_{n-1}}-\operatorname{Tr}_{A_{n}}[W_{n}],\ldots
	,W_{2}\otimes I_{B_{2}}-\operatorname{Tr}_{A_{3}}[W_{3}],\\
	W_{1}\otimes I_{B_{1}}-\operatorname{Tr}_{A_{2}}[W_{2}],W_{n+2}-W_{n+1}%
	,0,\ldots,0,0).
\end{multline}
Adding $\Phi_{1}^{\dag}(Y_{1})$ and $\Phi^{\dag}(Y_{2})$ gives%
\begin{multline}
	\Phi_{1}^{\dag}(Y_{1})+\Phi^{\dag}(Y_{2})=\text{diag}(-\operatorname{Tr}%
	[W_{n+2}\Gamma^{\mathcal{M}^{(n)}}],W_{n}\otimes I_{B_{n}}-W_{n+1},\\
	W_{n-1}\otimes I_{B_{n-1}}-\operatorname{Tr}_{A_{n}}[W_{n}],\ldots
	,W_{2}\otimes I_{B_{2}}-\operatorname{Tr}_{A_{3}}[W_{3}],\\
	W_{1}\otimes I_{B_{1}}-\operatorname{Tr}_{A_{2}}[W_{2}],W_{n+2}-W_{n+1}%
	+Z_{n}\otimes I_{B_{n}},\\
	Z_{n-1}\otimes I_{B_{n-1}}-\operatorname{Tr}_{A_{n}}[Z_{n}],\ldots
	,Z_{2}\otimes I_{B_{2}}-\operatorname{Tr}_{A_{3}}[Z_{3}], Z_{1}\otimes I_{B_{1}}-\operatorname{Tr}_{A_{2}}[Z_{2}]).
\end{multline}
Then the inequality $\Phi_{1}^{\dag}(Y_{1})+\Phi_{2}^{\dag}(Y_{2})\geq A$ is
equivalent to the following set of inequalities:%
\begin{align}
	-\operatorname{Tr}[W_{n+2}\Gamma^{\mathcal{M}^{(n)}}]  &  \geq-1,\\
	W_{n}\otimes I_{B_{n}}-W_{n+1}  &  \geq0,\\
	W_{n-1}\otimes I_{B_{n-1}}-\operatorname{Tr}_{A_{n}}[W_{n}]  &  \geq0,\\
	&  \vdots\nonumber\\
	W_{2}\otimes I_{B_{2}}-\operatorname{Tr}_{A_{3}}[W_{3}]  &  \geq0,\\
	W_{1}\otimes I_{B_{1}}-\operatorname{Tr}_{A_{2}}[W_{2}]  &  \geq0,\\
	W_{n+2}-W_{n+1}+Z_{n}\otimes I_{B_{n}}  &  \geq0,\\
	Z_{n-1}\otimes I_{B_{n-1}}-\operatorname{Tr}_{A_{n}}[Z_{n}]  &  \geq0,\\
	&  \vdots\nonumber\\
	Z_{2}\otimes I_{B_{2}}-\operatorname{Tr}_{A_{3}}[Z_{3}]  &  \geq0,\\
	Z_{1}\otimes I_{B_{1}}-\operatorname{Tr}_{A_{2}}[Z_{2}]  &  \geq0,
\end{align}
which can be rewritten as%
\begin{align}
	\operatorname{Tr}[W_{n+2}\Gamma^{\mathcal{M}^{(n)}}]  &  \leq1,\\
	W_{n}\otimes I_{B_{n}}  &  \geq W_{n+1},\\
	W_{n-1}\otimes I_{B_{n-1}}  &  \geq\operatorname{Tr}_{A_{n}}[W_{n}],\\
	&  \vdots\nonumber\\
	W_{2}\otimes I_{B_{2}}  &  \geq\operatorname{Tr}_{A_{3}}[W_{3}],\\
	W_{1}\otimes I_{B_{1}}  &  \geq\operatorname{Tr}_{A_{2}}[W_{2}],\\
	W_{n+2}+Z_{n}\otimes I_{B_{n}}  &  \geq W_{n+1},\\
	Z_{n-1}\otimes I_{B_{n-1}}  &  \geq\operatorname{Tr}_{A_{n}}[Z_{n}],\\
	&  \vdots\nonumber\\
	Z_{2}\otimes I_{B_{2}}  &  \geq\operatorname{Tr}_{A_{3}}[Z_{3}],\\
	Z_{1}\otimes I_{B_{1}}  &  \geq\operatorname{Tr}_{A_{2}}[Z_{2}].
\end{align}
The dual objective function is given by%
\begin{equation}
	-\operatorname{Tr}[B_{1}Y_{1}]-\operatorname{Tr}[B_{2}Y_{2}%
	]=-\operatorname{Tr}[Z_{1}]+\operatorname{Tr}[\Gamma^{\mathcal{N}^{(n)}%
	}W_{n+1}]-\varepsilon\operatorname{Tr}[W_{1}].
\end{equation}
So then the dual can be written as%
\begin{equation}
	\sup_{\substack{Z_{n},\ldots,Z_{1}\in\text{Herm},\\W_{n+2},\ldots,W_{1}\geq
			0}}-\operatorname{Tr}[Z_{1}]+\operatorname{Tr}[\Gamma^{\mathcal{N}^{(n)}%
	}W_{n+1}]-\varepsilon\operatorname{Tr}[W_{1}]
\end{equation}
subject to%
\begin{align}
	\operatorname{Tr}[W_{n+2}\Gamma^{\mathcal{M}^{(n)}}]  &  \leq1,\\
	W_{n}\otimes I_{B_{n}}  &  \geq W_{n+1},\\
	W_{n-1}\otimes I_{B_{n-1}}  &  \geq\operatorname{Tr}_{A_{n}}[W_{n}],\\
	&  \vdots\nonumber\\
	W_{2}\otimes I_{B_{2}}  &  \geq\operatorname{Tr}_{A_{3}}[W_{3}],\\
	W_{1}\otimes I_{B_{1}}  &  \geq\operatorname{Tr}_{A_{2}}[W_{2}],\\
	W_{n+2}+Z_{n}\otimes I_{B_{n}}  &  \geq W_{n+1},\\
	Z_{n-1}\otimes I_{B_{n-1}}  &  \geq\operatorname{Tr}_{A_{n}}[Z_{n}],\\
	&  \vdots\nonumber\\
	Z_{2}\otimes I_{B_{2}}  &  \geq\operatorname{Tr}_{A_{3}}[Z_{3}],\\
	Z_{1}\otimes I_{B_{1}}  &  \geq\operatorname{Tr}_{A_{2}}[Z_{2}].
\end{align}
Since all of the $Z$ variables are Hermitian, we can make the substitution
$Z_{i}\rightarrow-Z_{i}$ without affecting the optimal value. The final form
is then%
\begin{equation}
	\sup_{\substack{Z_{n},\ldots,Z_{1}\in\text{Herm},\\W_{n+2},\ldots,W_{1}\geq
			0}}\operatorname{Tr}[Z_{1}]+\operatorname{Tr}[\Gamma^{\mathcal{N}^{(n)}%
	}W_{n+1}]-\varepsilon\operatorname{Tr}[W_{1}]
\end{equation}
subject to%
\begin{align}
	\operatorname{Tr}[W_{n+2}\Gamma^{\mathcal{M}^{(n)}}]  &  \leq1,\\
	W_{n}\otimes I_{B_{n}}  &  \geq W_{n+1},\\
	W_{n-1}\otimes I_{B_{n-1}}  &  \geq\operatorname{Tr}_{A_{n}}[W_{n}],\\
	&  \vdots\nonumber\\
	W_{2}\otimes I_{B_{2}}  &  \geq\operatorname{Tr}_{A_{3}}[W_{3}],\\
	W_{1}\otimes I_{B_{1}}  &  \geq\operatorname{Tr}_{A_{2}}[W_{2}],\\
	W_{n+2}  &  \geq W_{n+1}+Z_{n}\otimes I_{B_{n}},\\
	\operatorname{Tr}_{A_{n}}[Z_{n}]  &  \geq Z_{n-1}\otimes I_{B_{n-1}},\\
	&  \vdots\nonumber\\
	\operatorname{Tr}_{A_{3}}[Z_{3}]  &  \geq Z_{2}\otimes I_{B_{2}},\\
	\operatorname{Tr}_{A_{2}}[Z_{2}]  &  \geq Z_{1}\otimes I_{B_{1}}.
\end{align}

\end{document}